\documentclass[a4paper]{amsart}
\pdfoutput=1
\usepackage{amsmath}
\usepackage{amssymb}
\usepackage{amsfonts}
\usepackage{amsthm}
\usepackage{amssymb}
\usepackage{amsbsy}
\usepackage{color}
\usepackage{graphicx}
\usepackage{subfigure}
\usepackage[utf8]{inputenc}
\usepackage{csquotes}
\usepackage[greek,english]{babel}

\usepackage{enumitem}
\usepackage[hyperref,doi,url=true,natbib,backend=biber,firstinits=true,sortcites=true,maxnames=5,sorting=nyt]{biblatex}
\RequirePackage[normalem]{ulem}
\usepackage[makeroom]{cancel}

\AtEveryBibitem{
  \clearlist{address}
  \clearfield{date}
  \clearfield{eprint}
  \clearfield{isbn}
  \clearfield{issn}
  \clearlist{location}
  \clearfield{month}
  \clearfield{series}
  \clearlist{language}
  \clearfield{urldate}
  \ifentrytype{unpublished}{}{
    \clearfield{url}
    \clearfield{urldate}}
  \ifentrytype{book}{}{
    \clearlist{publisher}
    \clearname{editor}
  }
}

\DeclareFieldFormat{urldate}{}

\usepackage[bookmarks]{hyperref}

\newtheorem{theorem}{Theorem}[section]
\newtheorem{corollary}[theorem]{Corollary}
\newtheorem{lemma}[theorem]{Lemma}
\newtheorem{proposition}[theorem]{Proposition}

\newtheorem{problem}[theorem]{Problem}
\theoremstyle{definition}
\newtheorem{definition}[theorem]{Definition}
\theoremstyle{remark}
\newtheorem{remark}[theorem]{\textbf{Remark}}
\newtheorem{example}[theorem]{Example}
\numberwithin{equation}{section}

\newcommand{\E}{\mathbb{E}}
\newcommand{\X}{\mathbb{X}}
\renewcommand{\P}{\mathbb{P}}
\newcommand{\Q}{\mathbb{Q}}
\newcommand{\D}{\Delta}
\newcommand{\Sphere}{\mathbb{S}}

\newcommand{\di}{\mathrm{d}}

\newcommand{\half}{\dfrac{1}{2}}

\newcommand{\R}{\mathbb{R}}
\newcommand{\Rp}{\R_+}

\newcommand{\as}{a.s.}

\newcommand{\ie}{i.e.}
\newcommand{\eg}{e.g.}
\newcommand{\cf}{c.f.}
\newcommand{\dx}{\mathrm{d}x}
\newcommand{\dy}{\mathrm{d}y}
\newcommand{\ds}{\mathrm{d}s}

\newcommand{\du}{\mathrm{d}u}
\newcommand{\dt}{\mathrm{d}t}

\newcommand{\pd}[2]{\frac{\partial #1}{\partial #2}}

\newcommand{\Fc}{\mathcal{F}}
\newcommand{\Gc}{\mathcal{G}}
\newcommand{\Hc}{\mathcal{H}}
\newcommand{\eps}{\varepsilon}
\newcommand{\indic}[1]{\boldsymbol{1}_{\{\ensuremath{#1}\}}}
\newcommand{\ind}{\boldsymbol{1}}

\newcommand{\T}{\mathrm{T}}

\DeclareMathOperator{\tr}{tr}

\newcommand{\Ep}[1]{\E\left[#1\right]}
\newcommand{\Eq}[1]{\E^{\Q}\left[#1\right]}

\renewcommand{\Pr}{\P}
\renewcommand{\vec}[1]{\mathbf{#1}}

\newcommand{\Nb}{\mathbb{N}}
\newcommand{\M}{\mathcal{M}}
\newcommand{\Pc}{\mathcal{P}}
\newcommand{\Ac}{\mathcal{A}}
\newcommand{\Oc}{\mathcal{O}}
\newcommand{\Bc}{\mathcal{B}}
\newcommand{\Wc}{\mathcal{W}}
\newcommand{\cadlag}{c\`adl\`ag}

\newcommand{\ol}{\overline}
\newcommand{\ul}{\underline}

\usepackage{xcolor}

\definecolor{orange}{rgb}{1,0.3,0.2}

\newcommand{\average}{\left(1-\boldsymbol{1}\cdot \pmb{\xi}^\alpha,\pmb{\xi}^\alpha\right)}

\newcommand{\Db}{\mathbb{D}}

\title[Model-independent bounds for Asian options]{Model-independent bounds for Asian options: a dynamic programming approach}

\author{Alexander M.~G.~Cox}
\thanks{Alexander M.~G.~Cox, Department of Mathematical
  Sciences, University of Bath, Bath, U.~K..\\ e-mail: \texttt{a.m.g.cox@bath.ac.uk}, web: \texttt{http://www.maths.bath.ac.uk/$\sim$mapamgc/}}
\author{Sigrid K\"{a}llblad}
\thanks{Sigrid K\"{a}llblad, CMAP, \'Ecole Polytechnique, Paris, France.\\ e-mail: \texttt{sigrid.kallblad@cmap.polytechnique.fr}}
\thanks{This project began while AC was visiting \'Ecole Polytechnique. SK gratefully acknowledges the financial support of the ERC 321111 Rofirm, the ANR Isotace, the Chairs Financial Risks (Risk Foundation, sponsored by Soci\'et\'e G\'en\'erale), Finance and Sustainable Development (IEF sponsored by EDF and CA)}

\date{\today}

\begin{document}
\begin{abstract}
  We consider the problem of finding model-independent bounds on the price of an Asian option, when the call prices at the maturity date of the option are known. Our methods differ from most approaches to model-independent pricing in that we consider the problem as a dynamic programming problem, where the controlled process is the conditional distribution of the asset at the maturity date. By formulating the problem in this manner, we are able to determine the model-independent price through a PDE formulation. Notably, this approach does not require specific constraints on the payoff function (\eg{} convexity), and would appear to generalise to many related problems.
\end{abstract}
\maketitle

\section{Introduction} \label{sec:Introduction}

Since the seminal paper of \citet{Hobson:1998aa}, there has been substantial interest in questions of the following form: given an asset with price $(S_t)_{t \in [0,T]}$, a derivative whose payoff, $X_T$, depends on the path of the asset, and the prices of call options at maturity time $T$, find a static portfolio of calls, and a dynamic trading strategy in the asset which superhedges the derivative at time $T$, under (essentially) \emph{any} model for the asset. The class of models considered are usually very large (for example, all models with continuous paths), and so the resulting price is usually called the model-independent superhedging price.

The problem of finding the model-independent superhedging price is closely related to the problem of identifying the largest model-based price: specifically, in a classical setting, one would expect the prices of all options to be given as the expected value under some risk-neutral measure,\footnote{For ease of presentation, we largely assume that the interest rate is zero; this has no substantial effect on our main results.} and by specifying the call prices at time $0$, the distribution of $S_T$ under this risk-neutral measure is determined. It is therefore natural to conjecture that the model-independent superhedging price is equal to $\sup_{\Q} \Eq{X_T}$, where the supremum is taken over all probability measures $\Q$ such that $(S_t)_{t \in [0,T]}$ is a martingale, and $S_T$ has the distribution determined by the call options. Recently a number of papers, starting with \citet{Beiglbock:2013ab} in discrete time, and followed up by \citet{Dolinsky:2013aa} in continuous time (see also \cite{Acciaio:2016aa,Bayraktar:2015aa,Beiglbock:2015aa,Biagini:2015aa,Dolinsky:2015aa,Hou:2015aa}), have made this result explicit under a variety of technical conditions. Note that in this formulation it is very natural to consider the supremum over the set of probability measure as a primal problem, and the infimum over the class of super-hedging strategies as the corresponding dual problem.

An alternative approach to these problems, following \citet{Hobson:1998aa}, is to use a time-change argument to reformulate the primal problem in terms of an optimisation over solutions to the Skorokhod embedding problem (SEP): that is, to argue that, up to an unknown time change $\tau_t$, the martingale $S_t = B_{\tau_t}$ is a time change of a Brownian motion. For a number of important quantities (maximum, quadratic variation, local time, \dots) the values of these quantities for the asset price up to time $T$ and for the Brownian motion up to the (stopping time) $\tau_T$ agree. It can often then be argued that the choice of a model for $S_t$ with given law, and the choice of a stopping time $\tau_T$ are equivalent provided $B_{\tau_T}$ has the required distribution (and satisfies an integrability constraint). The latter problem is well known as the Skorokhod embedding problem. A common approach to solving the model-independent superhedging problem is then to consider the corresponding Skorokhod embedding problem. If an optimal solution to this problem can be found, then it is often possible to guess the correct solution to the corresponding dual problem, and interpret this in terms of a superhedging strategy. This approach has been used in \eg{} \cite{Brown:2001aa,Cox:2008aa,Cox:2011aa,Cox:2015aa,Cox:2013ab,Cox:2013aa,Henry-Labordere:2016aa,Hobson:2002aa,Hobson:2013aa,Kallblad:2015aa}; see also the survey article of \citet{Hobson:2011aa}.

Of note in all of the known optimal solutions to the SEP is that some underlying structure is required on the form of the option payoff: for example, if we write $M_t := \sup_{u \le t} B_u$ for the maximum process, then the optimal constructions are known to maximise $\Ep{F(M_\tau)}$ over solutions to the SEP, provided that $F$ is monotonic. To the best of our knowledge, the optimal construction when $F$ is not monotonic is not known. Similarly, in the case of variance options or the local time, the function $F$ must be concave/convex in order to have a known optimal solution. In \citet{Beiglbock:2013aa}, this behaviour was explained in terms of a natural convexity property which holds when a path-swapping operation is performed. It follows from this operation that many constructions of solutions to the SEP are optimal when the payoff to be optimised has such a convexity property. However without the corresponding convexity, a `nice' description of the optimal solution seems impossible. One of the key strengths of the results described in this paper is that our methods are \emph{not} constrained by such a convexity assumption on the payoff, and therefore will work for very general payoffs.

The main results in this paper concern the case where the option described above is an Asian option, that is, $X_T = F(A_T)$, where $A_t = \int_0^t S_u \, \du$ (we omit the usual scaling factor, $\frac{1}{t}$ for notational ease), and we consider the primal version of the problem, that is, we look to maximise $\Ep{F(A_T)}$ over all price processes $(S_u)_{u \in [0,T]}$ which are martingales, and which satisfy a constraint on the terminal law, $S_T \sim \mu$. Notably, the Asian example already falls outside the case of payoffs which can easily be handled by SEP methods, since the whole time-change $(\tau_u)_{u \in[0, T]}$, and not just the final time, $\tau_T$, is already important in determining the value of $A_T$. However, in the case where the function $F$ is convex, the optimal model is still easily determined: essentially, the asset will jump to its terminal distribution immediately, and the manner in which this is done (the `embedding component') turns out to be irrelevant. This result was first given in \citet{Stebegg:14}, which, to the best of our knowledge, is the first paper to characterise optimality in a setting where the SEP approach fails, or more generally to consider a problem of this form in continuous time without using the SEP approach. The standing assumptions in \cite{Stebegg:14} are slightly different to ours --- essentially, \cite{Stebegg:14} allows a slightly more general setup (general starting measures, and discrete and continuously monitored payoffs are included) at the cost of considering only convex payoff functions (see also Section~\ref{sec:convex}). At a heuristic level, this restriction to convex functions in \cite{Stebegg:14} appears comparable to the convexity constraint described above for the SEP in determining the `simple' set of optimisers. We also observe that there is a long history of considering model-independent bounds for the prices of Asian options (\eg{} \cite{Dhaene:2002aa,Albrecher:2008aa,Albrecher:2005aa,Chen:2008aa,Deelstra:2008aa,Forde:2010aa}), although we note that, in contrast to the case considered in this paper, existing results tend to consider discretely monitored Asian options, often when call options on the underlying are traded at some or all intermediate maturities.

The novelty of our approach relates to the manner in which we formulate the problem as a dynamic programming problem. In particular, we include the conditional law of the final value of the asset price in the formulation of our problem. The condition that the process is a martingale with this conditional law is then formulated in terms of the behaviour of the conditional law. Specifically, we require the conditional law $\xi_t$ to be a \emph{measure-valued martingale}, by which we mean that $\left(\xi_t(A)\right)_{t \in [0,T]}$ is required to be a martingale for any (Borel) set $A$. We will show in Section~\ref{sec:MVM} that this condition is equivalent to the original formulation. In particular, by requiring $\xi_0 = \mu$ and requiring $\xi_T$ to be singular, we enforce the condition that the terminal law of $S_t = \int x \,\xi_t(\dx)$ is $\mu$. The concept of a measure-valued martingale is classical, (see \eg{} \citet{Dawson:1993aa,Horowitz:1985aa}; in this literature, the rather confusing terminology `martingale measure' is also common), and has appeared in the context of the SEP in \citet{Eldan:2013aa}. A key result for our purposes is that we are able to show that our value function is continuous in $\xi$, where the space of measures is equipped with the Wasserstein topology. This allows us to approximate $\xi$ by atomic measures, which enables us to reduce the whole problem to a finite-dimensional problem, at which point classical methods can be used (Section~\ref{sec:dynam-progr-probl}). We note that, in this discrete formulation, our problem could be compared to (a special case of) the problems considered in \citet{Zitkovic:2014aa,El-Karoui:2013aa,Bouchard:2015aa}, although we prove our results via more direct, classical methods. We also remark that \citet{Galichon:2014aa} have also previously used a stochastic control approach to solve a similar problem, but in a rather different manner to the approach of this paper.  In Section~\ref{sec:exampl-superh} we are able to use these results to provide concrete solutions to certain problems.

We believe that the methods and ideas we describe in this paper can be applied far beyond the case of Asian options. However, the Asian option setting does provide us with some useful structure which we are able to exploit in the construction and formulation of our problem. In particular, it is easy to show that `small' changes in the conditional terminal law result in small changes in the value function for the problem, the increase in the average, $\di A_t$, is easy to write in terms of the current conditional law, and also our underlying problem is not strongly affected by jumps in the process: particularly, the value function for the problem where the path is assumed to be continuous, and the problem where the path is assumed to be \cadlag{} are identical (although optimisers may exist in the \cadlag{} formulation, and not in the continuous formulation). In Section~\ref{sec:concl-furth-work} we discuss further extensions.

\section{Problem formulation using measure-valued martingales} \label{sec:MVM}

Consider the following problem: we have an asset $(S_t)_{t \in [0,T]}$ in a market with a riskless bank account and a time-horizon $T$, and we wish to find a model-independent super-hedge of an option which pays the holder $F(A_T)$, where $\frac{1}{T}A_T = \frac{1}{T}\int_0^T S_t \, \dt$ is the running average\footnote{We use the slightly unconventional notation $A_T = \int_{0}^T S_t \, \dt$ to avoid an unnecessary number of terms of the form $\frac{1}{T}$ in all our calculations; it is clear that this is just a scaling factor and can be removed.}. We will consider the problem where the law of the underlying asset at maturity, $S_T$, is given at time $0$, and we consider the primal optimisation problem: that is, to find the law of the process which maximises $\Ep{F(A_T)}$ subject to $S_T \sim \mu$. Here, we consider the case where the interest rate $\rho=0$, although the extension to constant interest rates is straightforward.

Our basic approach is to consider the problem as a dynamic programming
problem where the current state includes the conditional distribution
of the process at maturity. Specifically, we assume $S_T \in \Rp$, and with $\M(\Rp)$ the set of Borel measures on $\Rp$, we write
\begin{equation}
  \label{eq:M1Defn}
  \Pc^1:= \{ \mu \in \M(\Rp) : \mu(\Rp) = 1, \int |x| \, \mu(\dx) < \infty\}.
\end{equation}
Our aim is to set the problem up as a dynamic programming problem. We suppose
that the problem evolves on an artificial time horizon, $r\ge 0$, on which a
measure-valued process $(\xi_r)_{r \ge 0}, \xi_r \in \Pc^1$ evolves. We let $(T_r)_{r \ge 0}$ be an increasing
process in $[0,T]$. Our interpretation of this quantity is that $T_r$ represents
the `real' time at the artificial time $r$. Roughly, the slower $T_r$ increases,
the higher `volatility' we see in the real-time scale. We set the problem up in
this way, since we wish to allow a substantial change in the $r$ time-scale to
happen instantaneously in real time, which may correspond to jumps in the asset
price. However, we wish to maintain a `continuous' evolution of the
measure-valued process over its natural time-scale (we do not wish to deal with
jumps in the measure-valued process). The choice of the increasing processes
$(T_r)$ will form part of the control of the problem --- specifically, we optimise
over $\lambda_r \in [0,1]$ and define
\begin{equation}\label{eq:TCDefn}
  T_r = \int_{0}^r \lambda_s \, \ds.
\end{equation}

The second part of the control will be the choice of the measure-valued process $\xi$. This process will determine the conditional distribution of the asset $(S_t)$. Specifically, the initial value is $\xi_0 = \mu$, where $\mu$ is the terminal law of the asset at time $0$, and over time we suppose that $(\xi_r)$ evolves in a manner that ensures that $(S_t)$ remains a martingale.

\begin{definition}
  \label{def:MVM}
  We say that an adapted process $(\xi_r)_{r \ge 0}$ with $\xi_r\in \Pc^1$ is a {\it measure-valued
    martingale} if, for any $f \in C_b(\Rp)$, $\xi_\cdot(f):=\int f(x)\,\xi_\cdot(\dx)$ is a martingale.
\end{definition}

Note trivially that, if $f \in C_b(\Rp)$, then $\xi_r(f)$ is bounded for each $t$, and hence the martingale $\xi_\cdot(f)$ is uniformly integrable, with well defined limit $\xi_{\infty}(f)$ (in particular, $\xi_\infty$ is a measure; see \citep[Proposition~2.1]{Horowitz:1985aa}).

\begin{remark}
  \label{rem:mvm_def}
  An adapted process $(\xi_r)$ with $\xi_r \in \Pc^1$ is a measure-valued
  martingale if and only if $\xi_\cdot(A)$ is a martingale for any
  $A\in\Bc(\R)$. Indeed, the indicator function over an interval of
  $\R$ may be approximated by continuous functions, and an application
  of the monotone class theorem yields that the claim holds for any
  $A\in\Bc(\R)$ (see Lemma \ref{lem:exist_mvm} for a similar
  argument). Conversely, any $f\in C_b(\Rp)$ may be approximated from
  below by simple functions. In fact, by the same argument, we see
  that $\xi_\cdot(f)$ is a martingale for any (non-negative) measurable
  function. 
\end{remark}

\begin{remark} \label{rem:xi_version}
  Our underlying probability spaces will generally be assumed to satisfy the usual conditions. Under this assumption, of course, for every $f \in C_b(\Rp)$, the martingale $\xi_\cdot(f)$ has a \cadlag{} version. More pertinently, we can choose a version of $(\xi_r)$ such that $\xi_\cdot(f)$ is \cadlag{} for every bounded Borel function $f$, see \citep[Theorem~2.5]{Horowitz:1985aa}. In what follows, we will assume that we always take this version of $(\xi_r)$.
\end{remark}

We will think of measure-valued martingales as processes, evolving in time. We emphasise that the support of the measure-valued martingale can only ever decrease: if $\xi_{r_0}(A) = 0$ then $\xi_{r}(A) = 0$ for all $r \ge r_0$. In the particularly nice case that $\xi_{r_0}$ is an atomic measure, then for all $r \ge r_0$, the measure will also be atomic, and supported on the same set of points. In particular, the spatial distribution of such a measure will not change, only the weights attributed to each point. Since the weight associated to each point is a martingale and constrained to lie in the interval $[0,1]$, it follows that in the limit as $r \to \infty$, the weight assigned to each point must converge to a limit; commonly, this limit will be assumed to be either $0$ or $1$, and this motivates the following definitions. Consider the set of singular measures on $\Rp$, $\Pc^s = \{ \mu \in \M(\Rp) : \mu = \delta_y, y \in \Rp\}$, then:

\begin{definition}
  \label{def:TMVM}
  We say that a measure-valued martingale $(\xi_r)$ is {\it terminating}
  if $\xi_r \to \xi_\infty\in \Pc^s$ \as{} as $r\to\infty$, where the
  convergence is in the sense of weak convergence of measures. It is
  {\it finitely terminating} if $\tau_s:= \inf \{r \ge 0: \xi_r \in \Pc^s\}$ is almost surely finite.  \end{definition}

\begin{lemma} \label{lem:MVMProp}
  Suppose $(\xi_r)$ is a terminating measure-valued martingale with $\xi_0 = \mu$. Then $X_\cdot = \int x \, \xi_\cdot(\dx)$ is a non-negative UI martingale with $X_{\infty} \sim \mu$.  
\end{lemma} 
\begin{proof}
  The martingale property follows from Remark \ref{rem:mvm_def}. In particular,
  for $y \in \R_+$, we have that
  \begin{align*}
    \Ep{(y-X_{\infty})_+} &= \Ep{\int (y-x)_+ \, \xi_\infty(\dx)}
                           = \int (y-x)_+ \, \mu(\dx).
  \end{align*}
  Since $\Ep{(y-X_\infty)_+}$ characterises the law of $X_\infty$ uniquely,
  $X_\infty \sim \mu$. 
  As $X_0 = \int x \, \mu(dx) < \infty$, it also follows that $\Ep{|X_r|} = \Ep{X_r} < \infty$.
  Finally, we observe that $X_r = \Ep{X_\infty | \Fc_r}$ and
  $X_\infty \sim \mu \in \Pc^1$ imply $X$ is a UI martingale.
\end{proof}

\begin{corollary}
  If $(\xi_r)$ is a terminating measure-valued martingale with $\xi_0=\mu$,
  then for every $1$-Lipschitz function $f$, $X_\cdot^f := \xi_\cdot(f) = \int f(x) \xi_\cdot(\dx)$ is
  a uniformly integrable martingale with $X_0^f = \int f \, \di \mu$ and
  $X_\infty^f \sim f(\mu)$.
\end{corollary}

We also wish to discuss the continuity of the process $(\xi_r)$. In order to do
this, we make the following definition:
\begin{definition}
  \label{def:CMVM}
  We say that a measure-valued martingale is \emph{continuous} if, for any
  1-Lipschitz function $f$, $X_\cdot^f = \int f(x) \, \xi_\cdot(\dx)$ is almost surely continuous.
\end{definition}
It immediately follows that $X_\cdot = \int x \, \xi_\cdot(\dx)$ is a
continuous process whenever $(\xi_r)$ is continuous. This is also
equivalent to requiring (almost sure) continuity of $r \mapsto \xi_r$ in the topology of
$\Wc_1$, the first Wasserstein metric, by the duality of the
Wasserstein distance \citep[Theorem 6.1.1]{Ambrosio:2008aa}.

Having introduced these concepts, we will take the second control in our problem
to be the choice of a process $(\xi_r)$, subject to the constraint that $(\xi_r)$ is
a terminating, continuous measure-valued martingale with $\xi_0 = \mu$.

Observe that, once we have chosen a process $(\xi_r)$, the `asset price' at time $T_r$ is given by $\int x \, \xi_{r}(\dx)$. Since the process $T_r$ is non-decreasing, there exists a right-continuous inverse, $T^{-1}_t = \inf \{ r>0 : T_r > t\}$, and introduce $T^{-1,*}_{t} = \inf \{ r>0 : T_r\wedge T > t\}$; moreover, there can be only countably many jumps in $T^{-1}_t$. We therefore define the \cadlag{} process \begin{equation}\label{eq:SDefn}
  S_t = \int x \, \xi_{T^{-1,*}_t}(\dx), \qquad t \le T,
\end{equation}
and note that $S_T = \int x \, \xi_\infty(\dx)$. The average process is then given by
\begin{equation}
  \label{eq:AvDefn}
  A_t = \int_0^t S_s \, \ds = \int_0^t \int x \, \xi_{T^{-1,*}_s}(\dx) \, \ds.
\end{equation}

Then the main problem we wish to solve is the following:
\begin{problem}[Basic optimisation problem] \label{prob:basic} Given an integrable probability measure $\mu$ on $\Rp$ and a function $F:\Rp \to \Rp$, we want to find a probability space $(\Omega, \Hc, (\Hc_t)_{t \in [0,T]}, \P)$ satisfying the usual conditions, and a \cadlag{} UI martingale $(S_t)_{t \in [0,T]}$ on this space with $S_T \sim \mu$ which maximises $\Ep{F(A_T)}$ over the class of such probability spaces and processes.
\end{problem}

 \begin{remark}\label{rem:jump0}
   Since we do not require $\Hc_0$ to be trivial, $S_0$ need not be a constant. However, for the Asian option, it holds for any probability space and \cadlag{} martingale $S_t$ as given in Problem \ref{prob:basic}, that one may construct a sequence of \cadlag{} martingales $(S^n)$ such that $S^n_0=s_0\in\R$, $S^n_T\sim\mu$, and
  \begin{equation} \label{eq:conv_2}
    \lim_{n\to\infty}\E\left[F(A_T^n)\right]
    ~=~
    \E\left[F(A_T)\right].
  \end{equation} 
  Hence, the value of Problem \ref{prob:basic} remains the same under the additional assumption that $S_0=s_0$, and for any optimiser to the former problem an approximately optimal sequence may be constructed for the latter; we refer to Lemma 5.1 in \citet{Stebegg:14} for a precise argument (see also Assumption 3.9 in \citet{Guo:2015aa} and the proof of Lemma 4.1 in \citet{Dolinsky:2015aa} for related arguments).  We argue in the proof of Lemma \ref{lem:1} below that the value of Problem \ref{prob:basic} remains the same if restricting to martingales which are piecewise constant over arbitrary but finite partitions. Hence, a similar argument yields that the value of the problem also remains the same if we restrict to continuous martingales.

  To formalise this remark, and since we generally wish to work with probability spaces satisfying the usual conditions, we extend our framework slightly: given a complete probability space with a right-continuous filtration $(\Gc_t)_{t \ge 0}$, we can always extend the filtration to $(-\eps,\infty)$, for some $\eps>0$, by taking $\Gc_t$ to be the (completion of the) trivial $\sigma$-algebra for $t<0$. Similarly, a \cadlag{} process $Z_t$ on $[0,\infty)$ can be extended to a \cadlag{} process on $(-\eps,\infty)$ by setting $Z_t$ to be some constant value for $t<0$. Since this is constant we may write $Z_{0-}$ for this value without confusion. Similarly, to avoid the excessive use of $\eps$'s, we write $(\Gc_t)_{t \in [0-,\infty)}$ to denote a filtration extended in this manner. All other terminology (\eg{} martingales) are then to be understood in the obvious way.  \end{remark}

Our first claim is that Problem~\ref{prob:basic} is equivalent to the following formulation: \begin{problem}[Measure-valued martingale formulation] \label{prob:MVM} Given an integrable probability measure $\mu$ on $\Rp$ and a function $F:\Rp \to \Rp$, we want to find a probability space $(\Omega, \Gc, (\Gc_r)_{r\in[0-,\infty)}, \P)$ satisfying the usual conditions, a progressively measurable process $\lambda_r \in [0,1]$, such that $\int_0^r \lambda_s \, \ds \to \infty$ \as{} as $r \to \infty$, and a finitely terminating measure-valued $(\Gc_r)_{r \in [0-,\infty)}$-martingale $(\xi_r)_{r \in [0-,\infty]}$ with $\xi_{0-} = \mu$ and $\int x\,\xi_r(\di x)$ continuous a.s., which maximises $\Ep{F(A_T)}$ with $A_T$ given by \eqref{eq:AvDefn}.  \end{problem}

\begin{lemma}\label{lem:1}
  Problems~\ref{prob:basic} and \ref{prob:MVM} are equivalent, in the sense that the values coincide and if there exists an optimiser in Problem~\ref{prob:basic}, then we can construct an optimiser for Problem~\ref{prob:MVM}, and vice-versa; if the supremum for the problem can only be approximated, then equivalent approximating sequences can be found.
  
  Moreover, if $F$ is continuous, then the value of the problem remains the same if we restrict Problem \ref{prob:MVM} to probability spaces and processes such that the filtration $\Gc_r$ is the usual augmentation of the natural filtration of a $(\Gc_r)_{r \ge 0}$-Brownian motion and $(\xi_r)_{r \ge 0}$ is continuous in the sense of Definitions~\ref{def:TMVM} and \ref{def:CMVM}. \end{lemma}
	
  As a consequence, if the restricted measure-valued martingale problem admits a solution, then a corresponding optimiser may be constructed also for the Basic optimisation problem. 
  Before proving this result, we give an auxiliary lemma.   
  \begin{lemma}\label{lem:exist_mvm}
  	Suppose $(X_r)_{r\in[0,\infty]}$ is a martingale on $(\Omega,\Gc,(\Gc_r)_{r\in[0,\infty]},\P)$ such that $\E[|X_\infty|]<\infty$.
  	Then there exists a terminating measure-valued martingale, $(\xi_r)_{r\in[0,\infty]}$, such that $X_r=\int x\,\xi_r(\dx)$, \as{} for all $r\in[0,\infty]$. 
  \end{lemma}
  
  \begin{proof}
  	Define the $\Gc_\infty$-measurable random measure $\xi_\infty(\dx):=\delta_{X_\infty}(\dx)$. Then, $\xi_\infty\in\Pc^1$, a.s. 
  	Further, let $\Ac$ be a countable Boolean algebra generating $\Bc(\R)$ and define the $\Gc_r$-measurable set function $\xi_r$ by
 	\begin{eqnarray}\label{eq:cara_1}
  		\xi_r(A) :=\E[\xi_\infty(A)|\Gc_r],\quad A\in\Ac. 
  	\end{eqnarray} 	 
  Since $\xi_\infty$ is countably additive a.s., so is $\xi_r$. Indeed, for $A_n\in\Ac$, $n\in\Nb$, such that $\cup A_n\in\Ac$ and $A_i\cap A_j=\emptyset$, $i\neq j$, 
  		\begin{eqnarray*}
  			\xi_r\left(\cup A_n\right) 
  			= \E[\xi_\infty(\cup A_n)|\Gc_r]
  			= \sum_{n=1}^\infty \E[\xi_\infty(A_n)|\Gc_r]
  			= \sum_{n=1}^\infty \xi_r(A_n).   		
  		\end{eqnarray*}
                Since $\xi_r$ is also finite, it follows that \eqref{eq:cara_1} uniquely defines a $\Gc_r$-measurable measure on $\Bc(\R)$, up to a null set, on which we arbitrarily take $\xi_t = \delta_0$.
  	Next, let $\mathcal{O}:=\{A\in\Bc(\R): \xi_r(A)\textrm{ is a martingale on }[0,\infty]\}$. Since, for any $r\in[0,\infty]$, $\xi_r$ is a measure and thus continuous from below, it follows that $\Oc$ is a monotone class. Indeed, for $A_n\in\Oc$, $n\in\Nb$, with $A_n\subset A_{n+1}\subset...$, 
  		\begin{eqnarray*}
  			\E[\xi_\infty(\cup A_n)|\Gc_r] 
  			= \lim_{n\to\infty}\E[\xi_\infty(A_n)|\Gc_r]
  			= \lim_{n\to\infty}\xi_r(A_n)
  			= \xi_r(\cup A_n).
  		\end{eqnarray*}
  	Since $\Ac\subset\mathcal{O}$, we have by the monotone class theorem that $\xi_r(A)$ is a martingale for all $A\in\Bc(\R)$. 
  	Since $\E[\int|x|\,\xi_\infty(\dx)]=\E[|X_\infty|]<\infty$, this yields in particular that $\xi_r\in\Pc^1$, for $r\in[0,\infty]$. According to Remark \ref{rem:mvm_def}, $(\xi_r)$ is thus a measure-valued martingale. 
  	It is terminating by definition. It therefore follows directly from Lemma \ref{lem:MVMProp} that $\int x \, \xi_r(\dx)=X_r$, \as{} for $r\in[0,\infty]$. 
  \end{proof}  

\begin{remark} \label{rem:Horo}
  The above result can be partially found in \citep[Theorem~1.3]{Horowitz:1985aa}, on taking $(\xi_r)_{r \in [0,\infty]}$ as the optional projection of the random measure $\delta_{X_{\infty}}$. 
\end{remark}

  \begin{proof}[Proof of Lemma \ref{lem:1}] We show that every candidate solution to Problem~\ref{prob:basic} gives rise to a candidate solution to Problem~\ref{prob:MVM}, and \emph{vice-versa}. The claim about optimisers follows.

    We first suppose that we have a solution to Problem~\ref{prob:basic}. By \citet[Theorem~11]{Monroe:72}, there exists a probability space $(\Omega,\Gc', (\Gc_s')_{s \in [0,\infty)},\Pr)$, a $(\Gc_s')$-Brownian motion $(W_s)$ with $W_0 = \int x \, \mu(\dx)$, and a right-continuous $(\Gc_s')$-time change $(\tau_t)_{t \in [0,T]}$, such that $(S_t)$ and $(W_{\tau_t})$ are equal in law, $\tau_T$ is almost surely finite, and $W_{\cdot \wedge \tau_T}$ is a UI martingale.  
    We then define $\xi'_s$ to be the law of $W_{\tau_T}$ conditional on $\Gc_s'$. That is, we apply Lemma \ref{lem:exist_mvm} to the process $W_{\cdot\wedge\tau_T}$ to obtain a terminating measure-valued martingale $(\xi'_s)_{s\in[0,\infty]}$, such that $\int x \, \xi'_s(\dx) = W_{s\wedge \tau_T}$, a.s.
    Note that the properties of $\xi'_s$ are preserved by defining $\xi_{0-}'=\mu$, and that $\int x \, \xi'_s(\dx)$ must be continuous.

  We now need to construct a measurable process $\lambda_r$ giving rise to a time-change $T_r$ via \eqref{eq:TCDefn} such that the process $(S_t)$ given by \eqref{eq:SDefn} is the required process.  Note that by construction, $(S_t)_{t \ge 0}$ and $\left(\int x \, \xi'_{\tau_t}(\dx)\right)_{t \ge 0} = \left(W_{\tau_t}\right) _{t \ge 0}$ are equal in law, and therefore they both give rise to the same value of $\Ep{F(A_T)}$.  We will now modify the time-change and deduce that this gives rise to the correct process.  Specifically, we recall that $\tau_T$ is finite a.s., let
  \begin{equation} 
    \label{eq:def_Tinv_linear} 
    T^{-1}_t =  
    \begin{cases}
      \tau_{t}+\frac{t}{T} & \quad  t \le T\\
      \tau_T+1+\frac{(t-T)}{T} & \quad t >T
    \end{cases},
  \end{equation}
  and, in turn, define 
  $T_r:= \sup\{t \ge 0:T^{-1}_t\le r\}$.
 From \eqref{eq:def_Tinv_linear} we immediately see that $T^{-1}_t$ is strictly increasing, with $T^{-1}_t-T^{-1}_s \ge \frac{t-s}{T}$ for $t>s$, so that $T_r$ is non-decreasing and $\frac{1}{T}$-Lipschitz. In particular, $T_{\tau_T + 1 + r} = T+r$ for $r\ge 0$ so that $T^{-1}_t$ given by \eqref{eq:def_Tinv_linear} is indeed the right-continuous inverse of $T_r$.  Further, with $R_r=r-\frac{T_r\wedge T}{T}$ and $\xi_r:=\xi'_{R_r}$, $r\in[0,\infty]$, it follows that $R_{T^{-1}_t}(\omega)=\tau_t(\omega)$, $t<T$, and, thus, $(S_t)_{t < T}$ and $\left(\int x \, \xi_{T^{-1}_t}(\dx)\right)_{t < T}$ are equal in law. Indeed, $\xi'_r\in\Pc^s$ for $r\ge \tau_T$.  Therefore let $(\Gc_r)$ be the (right-continuous) time-changed filtration given by $\Gc_r=\Gc_{R_r}'$, $r\in[0,\infty)$.  Then $\xi_r$ is a finitely terminating measure-valued $(\Gc_r)$-martingale.  Further, $T^{-1}_t\in\Gc'_{\tau_t}=\Gc'_{R(T^{-1}_t)}=\Gc_{T^{-1}_t}$ and, thus, $T_r\in\Gc_r$. Recalling the properties of $T_r$, we deduce that there exists a process $\lambda_r \in [0,1]$ which is $\Gc_r$ measurable and such that $T_r= T \int_0^r \lambda_s \, \ds$. Hence (possibly by taking a modification), $\lambda_r$ can be assumed to be progressively measurable, and it is immediate that $T_r \to \infty$ as $r \to \infty$.

 To see the converse, suppose we are given a solution to Problem~\ref{prob:MVM}. From Lemma~\ref{lem:MVMProp} it follows immediately that $S_\cdot = \int x \, \xi_{T^{-1,*}_\cdot}(\dx)$ is the required process.

  It remains to argue the second part of the lemma. Indeed, in general, the time-change granted by \citet{Monroe:72} may not necessarily be measurable with respect to the Brownian filtration.  However, for any probability space $(\Omega, \Hc, (\Hc_t)_{t \in [0,T]}, \P)$ and \cadlag{} martingale $(S_t)_{t \in [0,T]}$, we may define a sequence $(S^n_t)_{t \in [0,T]}$ by
  \begin{equation*}
    S^n_t=S_{[nt/T]T/n},\quad n\ge 1. 
  \end{equation*}  	
  Then the $(S^n_t)$ are still martingales with $S_T\sim\mu$. Further, since $F$ is continuous, $F(A^n_T)$ converges a.s. to $F(A_T)$, and an application of Fatou's Lemma gives that $\E[F(A_T)]\le\liminf_{n\to\infty}\E[F(A^n_T)]$. In consequence, the value of Problem \ref{prob:basic} remains the same if restricting to martingales which are piecewise constant over arbitrary but finite partitions. Since any discrete martingale may be embedded in a Brownian motion with stopping times measurable with respect to the Brownian filtration (\cf{} \eg{} \cite{Dubins:1968aa}), it follows that we may restrict to Brownian filtrations $(\Gc_r)$ in Problem \ref{prob:MVM}. By the Martingale Representation Theorem, any $(\Gc_r)$-martingale is continuous. In consequence, recalling Remark~\ref{rem:xi_version}, the $(\xi_r)$ defined above can be assumed to be continuous in the sense of Definition~\ref{def:CMVM}. The fact that the resulting measure-valued martingale is finitely terminating, and that the first time the integral of $\lambda_s$ reaches $T$ is finite also follow immediately from this embedding procedure.  \end{proof}

\begin{remark}
  We note that the embedding of a process $(S_t)_{t\in[0,1]}$ into the
  pair of a continuous measure-valued martingale
  $(\xi_r)_{r \in [0,\infty]}$
  and time-change $(\lambda_r)_{r \in [0,\infty]}$,
  is not unique. In particular, choosing
  $T^{-1}_t:=1-e^{-\tau_t}+\frac{t}{T}$ (\cf{}
  \eqref{eq:def_Tinv_linear}) renders
  $T^{-1}_T\le 2$ a.s. and the problem might be viewed as evolving on
  the
  \emph{finite} time-scale $r\in[0,2]$. In Lemma
  \ref{lem:scaling_after_dynamics} below, we will consider yet another
  scaling
  which yields a specific relation between the evolution of the $\xi$
  and the
  $\lambda$. \end{remark}

\begin{remark} \label{rk:lambda1}
  We observe in fact that, from the proof of the lemma, if $\lambda_r = 1$ for $r\in[u,v)$, for some interval $[u,v)$, then $\xi_r=\xi_u$ for all $r\in[u,v)$. In particular, $\lambda = 1$ corresponds to a constant $\xi$ and, thus, $(S_t)$ is constant on $t\in[T_u,T_v)$.
\end{remark}

\section{The Dynamic Programming Problem} \label{sec:dynam-progr-probl}

\subsection{Problem formulation and continuity}

We want to write the optimisation problem as a `Markovian' optimisation problem:
we suppose that at time $r$, we have `real' time $T_r=t$, current law
$\xi_{r} = \xi \in \Pc^1$, running average $A_{T_r}=a$, and we wish to find:
\begin{equation} \label{eq:Udefn} 
	U(r,t,\xi,a) = \sup \Ep{F(A_T)|T_r = t, \xi_{r} = \xi, A_{T_r} = a},
\end{equation}
where the supremum is taken over all time-change determining processes $(\lambda_u)_{u \in [r,\infty)}$ and measure-valued martingales $(\xi_{u})_{u \in [r,\infty)}$ satisfying the conditions of Problem 2.10. By Lemma~\ref{lem:1} it follows that $U(0,0,\mu,0)$ will be the value of the Asian option under the optimal model. 
At this stage, we directly define the function in \eqref{eq:Udefn} to equal the value of Problem \ref{prob:MVM} when the law to be embedded is given by $\xi$, the horizon by $T-t$, and the payoff function by $F(a+\cdot)$.
Then we have:
\begin{lemma}\label{lem:Ucontinuity}
  Suppose $F$ is a non-negative, Lipschitz function. The function $U: \R_+ \times [0,T] \times \Pc^1 \times \Rp \to \R$ is continuous (here the topology on $\Pc^1$ is the topology derived from the Wasserstein-1 metric), and independent of $r$.
\end{lemma}

\begin{proof}
  We begin by observing that continuity in all the variables except $\xi$ follows immediately from the continuity properties of $F$: any small change in $a$ will result in a direct shift in the final value of $A_T$, while small changes in $t$ can be handled by computing the average of the same model over the modified time-horizon. In addition, the independence of the problem on the value of the `measure-valued' time-scale, $r$, is immediate.

  We consider finally the continuity in the measure, $\xi$.  Consider a given probability space $(\Omega,\Gc,(\Gc_s)_{s\ge r},\P)$ and a measure-valued martingale $(\xi_s)_{s\ge r}$. Recall that $\Wc_1$ is the Wasserstein-1 metric space, and write $d_{\Wc_1}$ for the metric on this space. Let $\xi'\in\Pc^1$. We will first show that, if $d_{\Wc_1}(\xi_r,\xi') < \eps$, then there is a measure-valued martingale $(\xi_s')_{s \ge r}$ such that $\xi_r'=\xi'$ and $\Ep{\left.\left|\int x \, \xi_{s}(\dx)-\int x \, \xi_s'(\dx)\right|\,\right|\Gc_r}< \eps$ for all $s \in [r,\infty)$. Recall that $d_{\Wc_1}(\xi_r,\xi') < \eps$ implies that there exists a transport plan, $\Gamma \in \M(\Rp \times \Rp)$ such that $\xi'(\dy) = \Gamma(\Rp,\dy)$, $\xi_r(\dx) = \Gamma(\dx,\Rp)$ and $\int \int |x-y| \,\Gamma(\dx,\dy) < \eps$. First, by disintegration (\eg{} \cite[Theorem~5.3.1]{Ambrosio:2008aa}) there exists a Borel family of probability measures, $m(x,\dy)$ such that $\Gamma(\dx,\dy) = \xi_r(\dx) \,m(x,\dy)$.
  
  Now define the process
  \begin{equation*}
    \xi_s'(\dy) := \int \xi_s(\dx) \, m(x,\dy), \quad s \ge r.
  \end{equation*}
  Then $\xi_s'\in\Pc^1$ and $\xi'_r=\xi'$.  Further, for any $A\in\Bc(\R)$, since $m(\cdot,A)$ is measurable,
  \begin{equation*}
    \Ep{\xi'_u(A)|\Gc_s}
    =\Ep{\int m(x,A)\xi_u(\dx)|\Gc_s}
    =\int m(x,A)\xi_s(\dx)
    =\xi'_s(A),\quad s\le u, 
  \end{equation*}
  and, thus, $\xi'_s$, $s\ge r$, is a measure-valued martingale.  Next, note that
  \begin{align*}
    \left| \int x \, \xi_{s}(\dx) - \int x \, \xi_s'(\dx)\right| 
    & = \left|\int \int (x-y) \,\xi_s(\dx) m(x,\dy)\right|\\
    & \le \int \int |x-y| \,\xi_s(\dx) m(x,\dy).
  \end{align*}
  Hence
  \begin{equation*}
    \Ep{ \left| \int x \, \xi_{s}(\dx) - \int x \, \xi_s'(\dx)\right| \Big|\Gc_r}
    \le \int \int |x-y| \,\Gamma(\dx,\dy).
  \end{equation*}
  By the definition of the metric on $\Wc_1$, since $d_{\Wc_1}(\xi_r,\xi_r')<\eps$, we can find a transport plan $\Gamma$ with the desired marginals and $\int \int |x-y| \,\Gamma(\dx,\dy) < \eps$. Fix some process $(\lambda_s)_{s \ge r}$, and write $A_t^{\xi, \lambda}$ for the average process corresponding to the measure-valued process $\xi$ and the time-change process $\lambda$, conditional on $\Fc_r$. Recalling \eqref{eq:AvDefn} we have 
  \begin{align*}
    \Ep{\left|A_T^{\xi,\lambda}-A_T^{\xi',\lambda}\right|\big|\Gc_r} 
    & =  \Ep{\left|\int_{T_r}^T \int x \, \xi_{T^{-1}_s}(\dx) \, \ds-\int_{T_r}^T \int x \, \xi_{T^{-1}_s}'(\dx) \, \ds\right|\Big|\Gc_r}\\
    & \le \Ep{\int_{T_r}^T \left|\int x \, \xi_{T^{-1}_s}(\dx) - \int x \, \xi_{T^{-1}_s}'(\dx)\right|\, \ds\Big|\Gc_r}\\
    & \le \eps (T-T_r).
  \end{align*}
  
  Now fix $\eps'>0$ and consider $\xi,\xi'\in\Pc^1$ such that $d_{\Wc_1}(\xi,\xi')<\eps'/(2T\zeta)$, where $\zeta$ is the Lipschitz constant of $F$. Then there exists $(\xi_s,\lambda_s)_{s \ge r}$, such that $\xi_r=\xi$ and $U(r,t,\xi,a)\le \Ep{\left.F\left(A_T^{\xi,\lambda}\right)\right|\Gc_r}+\eps'/2$. Using the estimate above, and by the Lipschitz property of $F$, we can moreover find $(\xi'_s)_{s \ge r}$, such that $\xi'_r=\xi'$ and $\Ep{\left.\left|F\left(A_T^{\xi,\lambda}\right)-F\left(A_T^{\xi',\lambda}\right)\right|\right|\Gc_r} \le \eps'/2$. It follows that
  \begin{equation*}
    U(r,t,\xi,a)
    \le \Ep{\left.F\left(A_T^{\xi,\lambda}\right)\right|\Gc_r}+\eps'/2 
    \le \Ep{\left.F\left(A_T^{\xi',\lambda}\right)\right|\Gc_r}+\eps' 
    \le U(r,t,\xi',a) + \eps'.
  \end{equation*}
  By symmetry, $|U(r,t,\xi',a)-U(r,t,\xi',a)|\le \eps'$. Finally, we note that joint continuity follows as a simple adaptation of this argument combined with the arguments for the other parameters.
\end{proof}

Since the function $U(r,t,\xi,a)$ is independent of the parameter $r$, we will often write $U(t,\xi,a)$ where there is no confusion.

\begin{remark}
  Continuity of the primal problem as a function of $\mu$ was proven by alternative methods in \citet[Theorem 4.1]{Dolinsky:2015aa}. As demonstrated in \citet[Proposition 4.3]{Guo:2015aa}, upper semi-continuity can be proven by yet an alternative method. We now recall their argument in the present context.
  To this end, consider the space of all \cadlag{} paths on $[0,1]$ and let the filtration be the one generated by the canonical process $(S_t)$. Problem \ref{prob:basic} can then be formulated as maximizing $\E[F(A_T)]$ over martingale measures satisfying the constraint $S_T\sim^{\Pr}\mu$.
   Given a sequence of probability measures $(\mu_n)$ on $\R_+$ converging in $d_{\Wc_1}$ to $\mu$, let $(\Pr_n)$ be a sequence of martingale measures such that $S_T\sim^{\Pr_n}\mu_n$, and 
  \begin{equation}\label{eq:conv_1} 
  	\lim_{n\to\infty}
    \E^{\P_n}\left[F(A_T)\right] ~=~ \limsup_{n\to\infty}
    U\left(0,\mu_n,0\right).
  \end{equation}
  According to \citet{Jakubowski:1997aa}, there exists a sub-sequence $(\P_{n_k})_{k\ge 1}$ which is weakly convergent with respect to the so-called $S$-topology on the set of \cadlag{} paths.
  Let $\Pr_0$ be the limiting measure. According to \citet{Guo:2015aa}, $\Pr_0$ is then a martingale measure and $S_T\sim^{\Pr_0}\mu$.
  Since the mapping $\omega\mapsto A_T(\omega)$ is $S$-continuous (\cf{} Corollary 2.11 in \cite{Jakubowski:1997aa}) it follows that
  \begin{equation*}
    U(0,\mu,0)
    ~\ge~
    \E^{\P_0}\left[F(A_T)\right]
    ~\ge ~
    \lim_{n\to\infty} \E^{\P_{n_k}}\left[F(A_T)\right],
  \end{equation*} 
  which combined with \eqref{eq:conv_1} yields the upper semi-continuity.
\end{remark}

\subsection{Reduction to a finite dimensional problem}

Our aim now is to provide a more concrete description of the function $U$. However, because the function $U$ is continuous in $\xi$, we can restrict ourselves to a nicer class of problems: specifically, we can approximate our object of primary interest, $U(t,\xi,a)$ by a sequence $U(t,\xi^N,a)$, where $\xi^N$ can be chosen to have nice properties. For our purposes, a natural simplifying assumption is to assume that the measures $\xi^N$ are atomic measures. In this case, as we shall see, the problem becomes much more tractable via classical methods. As a consequence of this reduction, we will be able to deduce that a Dynamic Programming Principle holds by standard results from the literature. However the more theoretical question of whether a DPP holds for the original formulation is proved in the appendix; this result will not be used elsewhere in the paper.

To do this, we let $\X_N = \{x_0,x_1, \dots, x_N\}$, where
$0 \le x_0 < x_1 < \dots < x_N$, and write $\Pc^1(\X_N) = \Pc^1 \cap \M(\X_N)$ and
$\Pc^s(\X_N) = \Pc^s \cap \M(\X_N)$. Observe that if $(\xi_r)$ is a terminating
measure-valued martingale and $\xi_0 \in \Pc^1(\X_N)$ then $\xi_r \in \Pc^1(\X_N)$ \as{}
for all $r \ge 0$ and $\xi_\infty = \delta_{x_i}$ for some $x_i \in \X_N$. Further, write
$\alpha \subseteq \{0, 1, \dots, N\}$, $\X_\alpha = \{x_i: i \in \alpha\}$, and
$\Pc^1(\X_\alpha), \Pc^s(\X_\alpha)$ etc. as above. In particular, $\X_N=\X_{\{0,1,...,N\}}$.

We then have the following characterisation:
\begin{lemma}\label{lem:xiAtomic}
  Suppose $\mu \in \Pc^1(\X_N)$.  Then $(\xi_r)$ is a measure-valued martingale
  with $\xi_0 = \mu$ if and only if $\xi^n_r:= \xi_r(\{x_n\})$ is a non-negative
  martingale for each $n$ and $\sum_{i=0}^{N} \xi^n_r = 1$. Moreover, $(\xi_r)$ is
  terminating if and only if $\xi^n_{\infty} = 0$ for all but one
  $n \in \{0,1,2,\dots, N\}$, almost surely, and $(\xi_r)$ is continuous if and
  only if $\xi_r^n$ is continuous for each $n$.
\end{lemma}

It is clear that there are similar statements where $\Pc^1(\X_N)$ is replaced by
$\Pc^1(\X_\alpha)$.

Then we consider the further consequence of Lemma~\ref{lem:1}: by the Martingale
Representation Theorem, working on the probability space granted by
Lemma~\ref{lem:1}, we can assume that the dynamics of $(\xi_r)$ are given by a
controlled Brownian motion, in a recursive formulation. For fixed $N \ge 1$, we
suppose that $(\xi_r)$ solves the SDE
\begin{equation}
  \label{eq:XiSDE}
  \di \xi_r^n = w^n_{r} \di W_r, 
\end{equation}
for $(W_r)$ a standard Brownian motion, and where
$\vec{w}_r = (w^1_r, \dots, w^N_r) \in \R^{N}$, $w^0_r=-\sum_{n = 1}^N w^n_r$,
and $\xi_r^n \in \{0,1\}$ implies $w^n_r = 0$, $n\in\{0,...,N\}$, ---
that is, as soon as one of the atoms disappears, it can never be resurrected.

Next, we show that $(\xi_r)$ and $(\lambda_r)$ may be chosen so that a specific
relation holds between $\vec{w}_r$ and $\lambda_r$.

\begin{lemma}\label{lem:scaling_after_dynamics}
  Let $\mu\in\Pc^1(\X_N)$ and consider a martingale $(S_t)_{t\in[0,T]}$
  represented via \eqref{eq:SDefn} by processes
  $(\lambda_r,\xi_r)_{r \in [0,\infty)}$ given on a probability space
  $(\Omega, \Gc, (\Gc_r)_{r \in [0,\infty)}, \P)$, such that $\lambda_r\in[0,1]$
  is a progressively measurable process and $(\xi_r)$ is a measure-valued martingale
  with $\xi_0= \mu$. Suppose further that $(\Gc_r)$ is the natural filtration of
  a Brownian motion $(W_r)$,
  $\inf \{r \ge 0: \int_0^r \lambda_s \, \ds = T\} < \infty$ \as{}, and $(\xi_r)$
  is continuous and finitely terminating. Then, w.l.o.g., we may assume that
  \begin{equation}\label{eq:condn}
    ||\vec{w}_u||^2 + \lambda_u = 1 - \indic{\xi_u \in \Pc^s} \indic{T_u = T}.
  \end{equation}
  That is, we can always choose a multiple
  $(\Omega, \Gc, (\Gc_r), \P,(\lambda_r,\xi_r))$ which represents $(S_t)$ via
  \eqref{eq:SDefn}, and which satisfies the above properties as well as
  \eqref{eq:condn}. 
\end{lemma}

\begin{proof} 
  Suppose $(\Omega, \Gc, (\Gc_r), \P,(\lambda_r,\xi_r))$ satisfy the assumptions of the lemma (apart from \eqref{eq:condn}). We aim to construct a time-change such that the time-changed filtration $(\overline\Gc_u)$ and time-changed processes $(\bar\lambda_u,\bar\xi_u)_{u \ge 0}$ satisfy the assertions. To this end, recall that $T_r$ is given by \eqref{eq:TCDefn} and let $\tau := \inf \{r: \xi_r \in \Pc^s \text{ and } T_r = T\}$; since $\xi_r$ is finitely terminating, $\tau$ is finite a.s. Let $\phi:\Omega\times[0,\infty)\to\Rp$ be given by $\phi(u)=\int_0^u\eta^2_s\di s$ for some positive, adapted process $\eta_u$ such that $\phi(\infty)\ge \tau$. Then $\phi$ is continuous and increasing in $u$, and its inverse $\phi^{-1}$ is well-defined and continuous on $[0,\tau]$. We define
  \begin{equation}\label{eq:defn}
    \bar\xi_u ~:= ~\xi_{\phi(u)}
    \quad\textrm{and}\quad
    \overline T_u~:=~T_{\phi(u)},
  \end{equation}	
  and let $(\overline\Gc_u)_{u\in[0,\infty)}$ the time-changed filtration with $\overline\Gc_u=\Gc_{\phi(u)}$. Note that $\bar\xi_u$ is a continuous measure-valued $(\overline\Gc_u)$-martingale. Moreover, \eqref{eq:defn} implies that $\overline T^{-1}_t =\phi^{-1} (T^{-1}_t)$, $t<T$ (recall that $\overline T^{-1}_T=\infty$). Hence, $S_t$ is given by \eqref{eq:SDefn} defined with respect to $\bar\xi_u$ and $\overline T_u$. It remains to argue that $\eta_u$ can be chosen such that 
  \begin{equation}\label{eq:condn_proof}
    ||\vec{\bar w}_u||^2 + \bar\lambda_u = 1 - \indic{\bar\xi_u \in \Pc^s} \indic{\overline T_u = T}.
  \end{equation}
  First, note that $\phi^{-1}(\tau) = \inf \{u: \bar\xi_u \in \Pc^s \text{ and } \overline T_u = T\}$.  By the Martingale Representation Theorem, we know that $\xi_r$ is given by \eqref{eq:XiSDE} for some process $(\vec{w}_r)\in \R^{N}$.  Since there is a $(\overline\Gc_u)$-Brownian motion $(B_u)$ such that $\di W_{\phi(u)}=\eta_u\di B_u$, it follows from \eqref{eq:defn} that
  \begin{equation}\label{eq:lambda_bar}
    \di \bar\xi^n_u~=~ w^n_{\phi(u)}\eta_u\di B_u
    \quad\textrm{and}\quad
    \bar\lambda_u~=~\eta^2_u \lambda_{\phi(u)}.
  \end{equation}
  Let $\vec{\bar w}_u=\vec{w}_{\phi(u)}\eta_u$. Then \eqref{eq:condn_proof} holds for $u\le \phi^{-1}(\tau)$, if $\eta$ satisfies
  \begin{equation*} \label{eq:def_eta}
    \eta^2_u~=~\frac{1}{||\vec{w}_{\phi(u)}||^2+\lambda_{\phi(u)}},\quad u\le
    \phi^{-1}(\tau).
  \end{equation*}
  We therefore proceed by defining $\phi^{-1}$ via
  \begin{equation*}
    \di \phi^{-1}(r):=\left(||\vec{w}_{r}||^2+\lambda_{r}\right)\di r, \qquad r\le \tau.
  \end{equation*}
  It follows from the construction of $(\xi_r)_{r\in[0,\infty)}$ and $(\lambda_r)_{r\in[0,\infty)}$, that for $r\le \tau$, $||\vec{w}_r||=0$ if and only if $\lambda_r=1$ (\cf{} Remark~\ref{rk:lambda1}, and note that $||\vec{w}||=0$, if and only if, $||(w^0,\vec{w})||=0$). In consequence, $\phi^{-1}$ is well-defined, continuous and strictly increasing on $[0,\tau]$. In particular,
  \begin{equation*}
    \phi^{-1}\left(\tau\right)
    ~=~ \int_0^{\tau}\left(||\vec{w}_{r}||^2+\lambda_{r}\right)\di r
    ~=~ T + \int_0^{\tau}||\vec{w}_{r}||^2\,\di r, 
  \end{equation*}
  and we observe that, as the quadratic variation process of a bounded martingale (in $\R^N$), $\int_0^{\tau}||\vec{w}_{r}||^2\,\di r$ is almost surely finite.  It follows that $\phi_u$ and $\eta_u$ are well-defined, for $u\le \phi^{-1}\left(\tau\right)$, and that $\phi_u$ and $\phi^{-1}_r$ are adapted with respect to the respective filtrations $(\overline\Gc_u)$ and $(\Gc_r)$.  In consequence, $\bar\xi_u$ and $\bar\lambda_u$ are well-defined via \eqref{eq:defn} for $u\le \phi^{-1}(\tau)$ and, according to \eqref{eq:lambda_bar}, $\bar\lambda_u=\eta^2_u\lambda_{\phi(u)}\in[0,1]$, and can therefore be assumed to be progressively measurable (possibly after taking a modification).  For $u>\phi^{-1}(\tau)$, we take $\vec{w}_u = 0 = \lambda_u$.
\end{proof}
\begin{remark}
  When embedding piecewise constant martingales as was done in the second part of the proof of Lemma \ref{lem:1} it follows that $\lambda_r\in\{0,1\}$. However, the solution to Problem~\ref{prob:MVM} (which in general is not unique), and thus to the basic optimisation problem, may be attained for more general processes $\lambda_u\in[0,1]$; \cf{} the non-convex example considered in Section \ref{sec:nonconvex}. Hence, we do not further restrict the set of $\lambda$'s even though the value of the problem would remain unaffected.  
\end{remark}

Given the above dynamics of the stochastic factors, we note that $U(r,t,\xi,a)$ in \eqref{eq:Udefn} is now well-defined as the value-function corresponding to a class of dynamic control problems. In particular, w.l.o.g., we may optimise over controlled processes defined on a fixed reference probability space; see \eg{} \cite{FlemingSoner}.	
The following result is now an immediate consequence of the lemma above.  Let
$\Ac^N_{u_0} = \{(\vec{w}_u)_{u \ge u_0}: \vec{w}_u\text{ prog. meas.}, \vec{w}_u =
(w^1_u, \dots, w^N_u) \in \R^{N} \textrm{ and } ||\vec{w}_u|| \le 1\}$
the set of admissible controls\footnote{Strictly speaking, we should also include here the
  set of possible probability spaces, as in \cite{FlemingSoner}; for ease of presentation, we omit this detail, which does not affect our arguments.}.

\begin{corollary} \label{cor:ControlForm}
  For each $\vec{w}\in \Ac^N_{u_0}$, define $(\lambda_u)_{u\ge u_0}$ by \eqref{eq:condn}, and $(\xi_u)_{u\ge u_0}$ by \eqref{eq:XiSDE} with $w^0=-\sum_{i = 1}^N w^i$.
  Then, for $\mu\in\Pc(\X_N)$,
  Problem \ref{prob:MVM} in its restricted form (\cf{} Lemma \ref{lem:1}) is equivalent to finding a process $\vec{w}\in\Ac^N_0$ such that $\xi_u^n \in \{0,1\}$ implies $w^n_s = 0$, $s\ge u$, for $n\in\{0,...,N\}$, and such that $\vec{w}$ maximises $\E[F(A_T)]$
  over the class of such processes where 
  \begin{equation}\label{eq:MarkovDynamics}
    \di A_{T_u} = (x_0,...,x_N)\cdot(\xi^0_u,...,\xi^N_u)\,\lambda_u T\, \du.
  \end{equation}
  Moreover, for all $\eps>0$, there exists $u^*=u^*(\eps)>0$ such that, for any $\mu, t, a$,
  \begin{equation*}
    \inf_{\vec{w}\in\mathcal{A}_{u_0}} \Pr(\xi_{u_0+u^*} \in \Pc^s(\X_N)|\xi_{u_0}=\mu,T_{u_0} = t,
    A_{T_{u_0}} = a) > 1-\eps.
  \end{equation*}
\end{corollary}
	
\begin{proof}
  The only part of the first half of the corollary that does not follow immediately from
  the previous result is that the process $(\xi_u)_{u \ge 0}$ which results from a given choice
  $\vec{w} \in \mathcal{A}_{u_0}$ is terminating, and this follows once we show the second
  half of the result.

  To see the second half of the result, note that it is sufficient to
  show that there is a similar bound for the first time that
  $\xi^n_u = 0$ for some $n \in \{0,1,\dots,N\}$. Consider the process
  at time $u^*\gg u_0$, and suppose that the measure $\xi$ has not
  already become singular at this time, so we have
  \begin{equation*}
    \int_{u_0}^{u^*} (||\vec{w}_u|| + \lambda_u)\, \du = u^*-u
    \implies \int_{u_0}^{u^*} ||\vec{w}_u|| \, \du \ge u^*-u-T.
  \end{equation*}
  In particular, we can ensure
  that $\max_{k} \left\{\int_{u_0}^{u^*} (w_u^k)^2 \, \du\right\}$ is
  arbitrarily large by choosing $u^*$ sufficiently large. Let $k^*$ be the maximising component; it follows
  immediately from the fact that $\xi_u^{k^*}$ is a
  $[0,1]$-valued martingale with quadratic variation process
  $\int_{u_0}^{u^*} (w_u^{k^*})^2 \, \du$, that with high probability at
  least one component must have become zero or one.
\end{proof}

Between Lemma~\ref{lem:Ucontinuity}, and
Corollary~\ref{cor:ControlForm}, we have shown that the problem
\eqref{eq:Udefn} is equivalent to choosing the variable $\vec{w}$ in
the problem above. Moreover, we can break the problem up into a
sequence of independent problems by considering the process only up to
the first time that one of the current atoms becomes zero. In
particular, for $\xi\in\Pc(\X^N)$, we can consider the problem:
\begin{equation}
  \label{eq:vtildedefn}
  \begin{split}
    \tilde{V}_N(u,t,\xi,a) = \sup_{\vec{w} \in \Ac^N_u}
    \E\Big[\tilde{V}_{N-1}&(\sigma,T_\sigma,\xi_{\sigma},A_{T_\sigma})
    \ind_{\{T_\sigma < T\}} \\ & {}+
    F(A_{T_\sigma})\ind_{\{T_\sigma = T\}}| A_{T_u} = a,
    \xi_{u} = \xi, T_u = t\Big],
  \end{split}
\end{equation}
where
$\sigma = \inf\{ s \ge u: \xi_{s}^n \not\in (0,1) \text{ some } n\in\{0,...,N\}
\text{ or } T_{s} = T\}$.
We also have the boundary conditions
$\tilde{V}_0(u,t,\xi,a) = F(a + (1-t) x)$, where $\xi = \delta_{x}$. Here, the function $\tilde{V}_{k}(u,t,\xi,a)$
is determined for $\xi \in \Pc^1(\X_\alpha)$ with $|\alpha| =
k+1$.
Specifically, for $\xi \in \Pc^1(\X_\alpha)$ with $|\alpha| = k+1$
\begin{equation*}
  \begin{split}    
  \tilde{V}_k(u,t,\xi,a) = \sup_{\vec{w} \in \Ac^k_u}
  \E\Big[\tilde{V}_{k-1}&(\sigma,T_\sigma,\xi_{\sigma},A_{T_\sigma})
  \ind_{\{T_\sigma < T\}} \\ & {} + F(A_{T_\sigma})\ind_{\{T_\sigma = T\}}|  A_{T_u} = a, \xi_{u} = \xi, T_u = t\Big],
  \end{split}
\end{equation*}
where we set
$\sigma = \inf\{ s \ge u: \xi_{s} \in \Pc^1(\X_\alpha) \text{ some } \alpha,
|\alpha| \le k \text{ or } T_{r} = T\}$.
Observe in particular that each $\xi \in \Pc^1(\X_N)$ determines a unique set
$\alpha$ such that $\xi \in \Pc^1(\X_{\alpha'})$ implies
$\alpha' \supset \alpha$. Write $\alpha(\xi)$ for this unique subset, and we
observe that we have the consistency conditions:
\begin{equation*}
  \tilde{V}_{|\alpha(\xi)|-1}(u,t,\xi,a) = \tilde{V}_k(u,t,\xi,a), \quad \text{ for all } k \ge |\alpha(\xi)|-1.
\end{equation*}
Finally, fix $\xi$ with $|\alpha(\xi)| = k+1$. We can make the identification
between the probability measure
$\xi = \sum_{j =0}^{k}\xi^{i_j} \delta_{x_{i_j}}$ (on
$\X_{\alpha}$), where $i_0, i_1, \dots, i_{k}$ are the ordered
elements of $\alpha$ and the vector
$\pmb{\xi}^\alpha = (\xi^{i_1}, \dots, \xi^{i_{k}}) \in
\D^{k}:= \{\vec{z} \in \R^{k}_+: \sum z_i \le1\}$.
Specifically, $\xi^{i_0}=1-\boldsymbol{1}\cdot \pmb{\xi}^\alpha$.
With this identification, we define:
\begin{equation*}
  V_{\alpha}(u,t,\pmb{\xi}^\alpha,a) = \tilde V_{k}(u,t,\xi,a).
\end{equation*}
We write
$\vec{x}^\alpha = (x_{i_0},x_{i_1},\dots,x_{i_{k}})$, and $\Sphere^k = \{\vec{z} \in \R^k: ||z|| = 1\}$ for the unit sphere in $\R^k$. 
Finally, note that for $|\alpha|=1$, $V_\alpha(t,a)=\tilde V_{0}(t,\xi,a)=F(a+(T-t)x_{i_0})$, and we then define $\average:=\xi^{i_0}=1$. We also use the convention $\Sphere^0:=\emptyset$ and $\sup \emptyset :=-\infty$.

\begin{theorem} \label{thm:dpp} Suppose $F(a)$ is continuous and non-negative.  Fix $\alpha\subseteq \{0,\dots, N\}$, with $|\alpha|\ge 1$, and write
  $k = |\alpha|-1$. The function
  $V_{\alpha}:\R_+ \times [0,T]\times \D^k \times \R_+ \to \R$ is independent of
  $u$, and is the unique non-negative viscosity solution bounded by $F(a+(T-t)x_N)$, to
  \begin{equation}\label{eq:dpp}
    \max\left\{\pd{V_{\alpha}}{t} + \vec{x}^\alpha \cdot \average
    \pd{V_{\alpha}}{a}, \sup_{\vec{w} \in \Sphere^k} \left[\tr(\vec{w} \vec{w}^\T D^2_{\pmb{\xi}}V_{\alpha})\right]\right\} = 0
  \end{equation}
  for $\pmb{\xi}^\alpha \in (\D^k)^\circ$, and $t < T$, with the boundary conditions
  \begin{equation}\label{eq:boundary_cond}
  \begin{array}{lll}
    V_{\alpha}(u,T,\pmb{\xi}^\alpha,a) &= & F(a) \\
    V_{\alpha}(u,t,\pmb{\xi}^{\alpha'},a) & = & V_{\alpha'}(u,t,\pmb{\xi}^{\alpha'},a)
  \end{array}
  \end{equation}
  where the second equation is taken when $\pmb{\xi}^\alpha \in \partial \D^k$. Here $\alpha'$ is the subset of
  $\alpha$ corresponding to non-zero entries of $\average$, and
  $\pmb{\xi}^{\alpha'}$ is the vector identifying the corresponding probability measure.
\end{theorem}

\begin{proof} 
  We work by induction; suppose the problem has been solved for
  $k'<k$, to give continuous value functions, independent of time. 
  The case where $k = 0$ is trivial. 
  The first step is to approximate by
  a problem with a finite time-horizon. To this end, we fix a sequence
  $K \nearrow \infty$, and
  consider the following two problems: For given $(u,t,\xi,a)$ with
  $\xi \in \Pc^1(\X_\alpha)$ and $|\alpha|=k+1$, we set
  $\sigma^K = \sigma \wedge (K+1)$ and define the functions
  $\tilde{V}^{\overline K}_k$ and $\tilde{V}^{\underline K}_k$ by
    \begin{equation*}
      \begin{split}
        \tilde{V}^{\overline K}_k(u,t,\xi,a)  = & \sup_{\vec{w}
          \in \Ac^{k,K}_u}
        \E\Big[\left(F(a+(T-t)x_N)(\sigma^K-K)_+\right) \vee 
        \left(F(A_{T_{\sigma^K}})\ind_{\{T_{\sigma^K} = T\}}\right.
        \\ & \left. ~+~ 
        \tilde{V}_{k-1}(\sigma^K,T_{\sigma^K},\xi_{\sigma^K},A_{T_{\sigma^K}})
        \ind_{\{T_{\sigma^K} < T\}}\right)
        \big|T_u = t, \xi_u = \xi, A_{T_u} = a\Big],
      \end{split}
    \end{equation*}
  and 
    \begin{equation*}
      \begin{split}
        \tilde{V}^{\underline K}_k(u,t,\xi,a)  = & \sup_{\vec{w}
          \in \Ac^{k,K}_u}
        \E\Big[\left(F(a+(T-t)x_N)(K+1-\sigma^K)_+\right) \wedge
        \left(F(A_{T_{\sigma^K}})\ind_{\{T_{\sigma^K} = T\}}\right.
        \\ & \left. ~+~ 
        \tilde{V}_{k-1}(\sigma^K,T_{\sigma^K},\xi_{\sigma^K},A_{T_{\sigma^K}})
        \ind_{\{T_{\sigma^K} < T\}}\right)
        \big|T_u = t, \xi_u = \xi, A_{T_u} = a\Big],
      \end{split}
    \end{equation*}
    where
    $\Ac^{k,K}_u = \big\{(\vec{w}_s)_{s \in [u,K+1]}:
    \text{ prog. meas. with }\vec{w}_s \in \R^{k}\textrm{ and }
    ||\vec{w}|| \le 1\big\}$
    and, as previously, 
    $\sigma = \inf\left\{ s \ge u: \xi_s^n \not\in (0,1) \text{ some }
      n\in\{0,...,k\} \text{ or } T_{s} = T \right\}$.
    With the same identification as above, we define $V^{\overline K}_{\alpha}$
  and $V^{\underline K}_{\alpha}$ by
  \begin{equation*}
    V^{\overline K}_{\alpha}(u,t,\pmb{\xi}^\alpha,a) = \tilde{V}^{\overline K}_{|\alpha(\xi)|-1}(u,t,\xi,a)
    \quad\textrm{and}\quad
    V^{\underline K}_{\alpha}(u,t,\pmb{\xi}^\alpha,a) = \tilde{V}^{\underline K}_{|\alpha(\xi)|-1}(u,t,\xi,a).
  \end{equation*}
 
  Recall that the dynamics of the involved factors is governed by \eqref{eq:TCDefn}, \eqref{eq:XiSDE} and \eqref{eq:MarkovDynamics}, with $\lambda_s$ given by \eqref{eq:condn}. Note that without loss of generality, we may write $\Ac^{k,K}_u = \big\{(\lambda_s,\vec{w}_s)_{s \in [u,K+1]}: (\lambda_s,\vec{w}_s)
    \text{ prog. meas. with }\vec{w}_s \in \Sphere^{k}\textrm{ and }\lambda_s\in[0,1]\big\}$.
    It follows from \cite[Corollary~V.3.1]{FlemingSoner} that on the domain
  $[0,K+1]\times[0,T]\times\R^k\times\R$, $V^{\overline K}_{\alpha}$ and $V^{\underline K}_{\alpha}$
  are both viscosity solutions to
  \begin{equation}\label{eq:hjb2}
    \pd{V_{\alpha}}{u} - \sup_{\substack{\vec{w}\in\Sphere^k,\\\lambda\in[0,1]}}
    \left[ \half (1-\lambda) \tr(\vec{w} \vec{w}^T D_{\pmb{\xi}}^2V_\alpha) + \lambda \left(\pd{V_{\alpha}}{t} + \vec{x}^{\alpha} \cdot \average\pd{V_\alpha}{a}\right)\right] = 0,
  \end{equation}
  equipped with the boundary conditions
  \begin{equation} \label{eq:cond_hjb} 
    \left\{
      \begin{array}{lll}
        V_{\alpha}(u,T,\pmb{\xi}^\alpha,a) &=& F(a) \\
        V_{\alpha}(u,t,\pmb{\xi}^{\alpha'},a)  &=& V_{\alpha'}(u,t,\pmb{\xi}^{\alpha'},a)
      \end{array}
    \right.
  \end{equation}	
  for $u<K$, and either increasing to $F(a+(T-t)x_N)$ for $u \in [K,K+1]$ in the first case, or decreasing to $0$ in $[K,K+1]$ in the second case. In both cases, we have a viscosity equation with controls in a compact set, and with continuous boundary data on a compact domain. It follows that both equations have unique, continuous viscosity solutions, and the viscosity solutions to both equations correspond to the value functions of the corresponding optimal control problems. In particular, we see immediately that $V_\alpha^{\ol{K}}(u,t,\pmb{\xi},a) \ge V_\alpha(u,t,\pmb{\xi},a) \ge V_\alpha^{\ul{K}}(u,t,\pmb{\xi},a)$ for $u \le K+1$. Moreover, from Lemma~\ref{lem:Ucontinuity}, identifying $U$ and $V_{\alpha}$, we know the function $V_{\alpha}$ is continuous, and moreover, from Corollary~\ref{cor:ControlForm}, we know that $V_{\alpha}^{\ol{K}}(u,t,\pmb{\xi},a)$ will decrease pointwise to $V_\alpha(u,t,\pmb{\xi},a)$ as $K \to \infty$, and $V_\alpha^{\ul{K}}(u,t,\pmb{\xi},a)$ will increase pointwise to the same limit. We conclude that $V_\alpha$ is a viscosity solution to the given equation (see \eg{} \citet{Barles:1991aa}).

  Now suppose that $W$ is another viscosity solution to the same equation, also non-negative and bounded by $F(a+(T-t)x_N)$. By the comparison principle, for any $K$, $V_\alpha^{\ol{K}}(u,t,\pmb{\xi},a) \ge W(u,t,\pmb{\xi},a) \ge V_\alpha^{\ul{K}}(u,t,\pmb{\xi},a)$, for $u \le K$. Hence $V_{\alpha}\ge W \ge V_{\alpha}$; that is, $W=V_\alpha$.
  Finally, we observe that the solution $V_\alpha$ is independent of $u$, by Lemma~\ref{lem:Ucontinuity}, so $\pd{V_{\alpha}}{u} = 0$, and optimising over $\lambda$ immediately gives the equivalent formulation. 
\end{proof}

\begin{remark}
  We note that some obvious generalisations of this setup can easily be handled. For example, consider Asian options with non-constant weighting, so $\tilde{A}_T = \int_0^T f(t) S_t \, \dt$, for some (possibly signed) continuous function $f:[0,T]\to \R$. A simple modification to the arguments above gives the same result with the corresponding PDE:
  \begin{equation*}
    \max\left\{\pd{V_{\alpha}}{t} + f(t)\vec{x}^\alpha \cdot \average
    \pd{V_{\alpha}}{a}, \sup_{\vec{w} \in \Sphere^k} \left[\tr(\vec{w} \vec{w}^\T D^2_{\pmb{\xi}}V_{\alpha})\right]\right\} = 0.
  \end{equation*}

\end{remark}

\section{Examples and Superhedging}\label{sec:exampl-superh}

In this section we consider some simple cases where explicit solutions to the viscosity equations in Theorem~\ref{thm:dpp} can be given. We also give some arguments regarding the construction of superhedging strategies. A number of the results in this section can be compared to the recent work of \citet{Stebegg:14}, but we emphasise that our results extend beyond the case where $F$ is convex, and we will consider such an example below.

\subsection{Convex payoff functions}\label{sec:convex}

\begin{lemma}\label{lem:convex}
  Suppose the function $F$ is convex and Lipschitz. Then for all $\xi \in \Pc^1(\Rp)$:
  \begin{equation*}
    U(t,\xi,a) = \int F\left(a+(T-t) x\right) \,\xi(\dx).
  \end{equation*}
  Moreover, an optimal model is given by:
  \begin{align*}
    S_{0-} & = \int x \, \xi(\dx)\\
    S_t & = S_T, \qquad \qquad  t \ge 0,
  \end{align*}
  where $S_T \sim \xi$.
\end{lemma}
In terms of the class of models considered in Corollary~\ref{cor:ControlForm}, this corresponds to a
model which takes $\lambda_u = 0$ until the measure $\xi_u \in \Pc^s$, and then setting $\lambda_u =
1$ until $T_u = T$.

\begin{proof}
  By continuity, we are only required to check that \eqref{eq:dpp} holds for atomic $\xi$. However, if we 
  write $\ol{\xi} = \int x \, \xi(\dx)$, then
  \begin{align*}
    \pd{U}{t} + \ol{\xi} \pd{U}{a} & = \int F'\left(a+(T-t)x\right) \left(\ol{\xi}-x\right)\,\xi(\dx) \\
                                   & \le \int F'\left(a+(T-t)\ol{\xi}\right)
                                     \left(\ol{\xi}-x\right)\,\xi(\dx) = 0.
  \end{align*}
  Moreover, if $t = T$ or $\xi \in \Pc^s$ then we have equality.
  
  Recalling the notation of Theorem~\ref{thm:dpp}, we have
  \begin{align*}
    U(t,\xi,a) = \sum_{j=0}^{|\alpha(\xi)|-1} F\left(a + (T-t) x_{i_j}\right)\xi^{i_j},
  \end{align*}
  and computing the second derivatives, we have $D^2_{\pmb{\xi}}U = 0$. Hence \eqref{eq:dpp} holds.
\end{proof}

In this convex case, we are easily able to provide a martingale inequality interpretation of this result. Indeed, this has already appeared in \cite{Stebegg:14}. Since this will help our intuition, we provide an alternative approach to \cite{Stebegg:14}, which will enable us to gain insight into the optimal strategies for the non-convex case. We restrict first to the case where $F(a) = (a-K)_+$, for some $K>0$, and we write $Y_t = A_t + (T-t) S_t$. We suppose also that $(S_t)$ is a continuous semi-martingale (although a pathwise analogue of this argument is possible, where $S_t$ is assumed just to have continuous paths). From the definition of local time, we get:
\begin{equation*}
  (A_T-K)_+ =  (Y_T-K)_+ = (Y_0-K)_+ + M_T + L_T^{Y,K},
\end{equation*}
where $M_T$ is a local martingale, and $L_T^{Y,K}$ is the local time of $Y$ at the level $K$. It
follows from the definition of $Y$, that we have:
\begin{equation*}
  L_{T}^{Y,K} = \int_0^T (T-t) \,\di L_t^{S,K_t}, \quad \text{ where } K_t = \frac{Y_t-A_t}{T-t},
\end{equation*}
so $L^{S,K_t}$ is the local time of the asset price along the curve $K_t$. That is, we have a local time contribution coming from the crossing of the curve $K_t$ by the asset price. However, for a given distribution of $S_T$, the local time at each value of $x$ is fixed, and decreases as $|x-S_0|$ increases.  We now argue that $L_{T}^{Y,K}$ is maximised by trying to accumulate all the local time on the curve $K_t$ as close as possible to time zero: that is, all crossings of $S_t = K_t$ should happen as close to time zero as possible. This happens since if $S_t \neq K_t$, then $|S_t-K_t|$ is increasing, and there will be less local time which can later be accumulated at $K_t$ since the process needs to accumulate the local time at a (relatively) more distant point. In addition, the factor $(T-t)$ which appears in the integral only makes the weight of local time accumulated at later times smaller.

It follows (and again, this can be made rigorous) that the optimal model should make all crossings of $K_t$ necessary to embed in a short time interval. After this time, it is irrelevant how the process behaves, so long as it either remains above or below $K_t$.

\begin{remark}\label{rk:posint}
  The cases where there is a positive interest rate can be handled similarly (the process $Y_t =A_t+\frac{S_t}{\rho}\left(e^{\rho(T-t)}-1\right)$ should be used instead). In addition, by adding constraints, one can extend to general convex functions, with the model which crosses each relevant curve $K_t$ corresponding to a convexity point of $F$ immediately being the optimal choice. 
\end{remark}

\subsection{A non-convex example}\label{sec:nonconvex}

In this section, we provide a solution to the problem for a non-convex example. Specifically, we use
the intuition from the convex case established above to try and find a solution to the
problem for a payoff function of the form:
\begin{equation}\label{eq:payoff_nonconvex}
  F(A_T)=(A_T-K_1)_+-(A_T-K_2)_+, \qquad K_1 < K_2.
\end{equation}
Given certain additional assumptions on the measure we wish to embed, we will then verify that an optimal model may be determined through the use of Theorem~\ref{thm:dpp}. We observe that the results of this paper simply verify the existence of an optimal model. Given the existence of an optimal model, the existence of a super-hedging strategy follows from general results (\eg{} \citet{Dolinsky:2013aa}).

The intuition established above suggests that we wish to gain the benefit of the convexity at $K_1$ immediately, while leaving the concavity at $K_2$ until as late as possible. However there is a trade-off, since the process may sacrifice some of the convexity at $K_1$ by waiting at $K_2$.  To specify this, note that since the payoff is constant for $A_T\ge K_2$, it must be suboptimal to have positive support on events for which $Y_t>K_2$, $t\in(0,T]$. Indeed, the payoff will not be improved by this but the martingale property of $Y$ implies that more mass must then be put on events yielding an average strictly less than $K_2$. In consequence, at least for some values of $K_1, K_2$, it is natural to conjecture that at time $0$, $S$ will either run to $K_2$, or to some lower level; at the lower level, the paths will behave as indicated by Lemma~\ref{lem:convex}.

For a measure $\mu$ with continuous support, we therefore define the level $\eta$ by
\begin{equation}\label{eq:nonconvex_eta}
  \eta:=\inf\left\{x\in\R: \int_\eta^\infty x\mu(\di x)\ge K_2\right\}. 		
\end{equation}
We then expect to accumulate all mass above $x=\eta$ into a branch of the underlying taking the value $S_t=K_2$, $t\in(0,T)$, and embedding $\ind_{x\ge \eta}\mu(\di x)$ at $t=T$.  As for the mass to be embedded on $[0,\eta)$, we expect the same optimal behaviour as detected for the convex case in Section \ref{sec:convex}.  Put differently, at $u=T^{-1}_0$, with probability $\int_\eta^\infty\mu(\di x)$ we expect the measure-valued martingale $\xi_u$ to take the value $\frac{\ind_{x\ge \eta}\mu(\di x)}{\mu([\eta,\infty))}$ and stay constant until $T^{-1}_T$, and with probability $\int_0^\eta\mu(\di x)$ we expect to recover the structure of Lemma \ref{lem:convex}.

To specify this, we restrict to a certain class of measures $\mu$. Specifically, we consider the problem at time $t \in [0,T]$ with current average $A_t = a$ when we take $\vec{x}^\alpha = (-1,0,1)$, so $|\alpha| = 3$, and consider the terminal distribution
\begin{equation}\label{eq:mu_nonconvex}
  \mu=(1-\beta-\gamma)\delta_{-1}+\beta\delta_{0}+\gamma\delta_{1}, \qquad \beta,\gamma\in(0,1).
\end{equation} 
	That is, $\pmb\xi^\alpha =(\beta,\gamma)$, and we write $V(t,a;\beta,\gamma)=V_\alpha(t,a;\pmb\xi^\alpha)$. Further, we let $K_1\in(-1,1)$ and $K_2\in(0,1)$. Suppose now that $a+\frac{\gamma}{\gamma+\beta}(T-t)<K_2\le a+(T-t)$. That is, the expected averages considering the mass at both $x=0$ and $x=1$, and at $x=1$ only, are, respectively, smaller and greater than $K_2$.  Following the reasoning above, at $u=T^{-1}_t$, we then expect to have split $\xi_u$ into the two measures:
\begin{equation} \label{eq:distr:iii}
  \xi^1=\frac{\bar\eta\delta_{0}+\gamma\delta_{1}}{\bar\eta+\gamma}
  \quad\textrm{and}\quad
  \xi^2=\frac{\big(1-\gamma-\beta\big)\delta_{-1}+\big(\beta-\bar\eta\big)\delta_{0}}{1-\gamma-\bar\eta},
\end{equation}
where (\cf{} \eqref{eq:nonconvex_eta}) $\bar\eta$ is given by
\begin{equation*}
  \bar\eta = \sup\left\{y: \frac{\gamma}{\gamma+y} \ge \frac{K_2-a}{T-t }\right\}
  = \gamma\left(\frac{T-t }{K_2-a}-1\right).
\end{equation*}
If $a<K_1$, this yields $V(t,a;\beta,\gamma)=(\gamma+\bar\eta)(K_2-K_1)$. However, if $a-(T-t)<K_1\le a$, the result for the convex case guides us to further split the measure $\xi^2$ into $\delta_{-1}$ and $\delta_0$ (equivalently, all mass at $x=0$ and $x=1$ might be accumulated in one measure; see further discussion below) and it follows that $V(t,a;\beta,\gamma)= (\gamma+\bar\eta)(K_2-K_1)+(\beta-\bar\eta)(a-K_1)$.  Similar considerations for the other cases guides us to define the following candidate value function:
\begin{equation} \label{eq:ex_v} 
  V(t,a;\beta,\gamma):=\left\{
    \begin{array}{lll}
      K_2-K_1 & \text{(i)} & K_2\le a^{-101}\\
      (2\gamma+\beta-1)(T-t)+a-K_1 & \text{(ii)} & K_1\le a^{-1}, a^{-101}<K_2\\
      \frac{2\gamma+\beta}{1+\frac{K_2-a}{T-t }}(K_2-K_1) & \text{(iii)} & a^{-1}<K_1, a^{-101}<K_2\le a^{01}\\
      \gamma(T-t )-(\gamma+\beta)(K_1-a) & \text{(iv)} & a^{-1}<K_1\le a^0, a^{01}<K_2 \\
      \gamma\frac{T-t }{K_2-a}(K_2-K_1) & \text{(v)} & a^0<K_1, a^{01}<K_2\le a^1\\
      \gamma(T-t-(K_1-a)) & \text{(vi)} &  a^0<K_1\le a^1<K_2 \\
      0 & \text{(vii)} & a^1<K_1
    \end{array}\right.  
\end{equation}  
where we used the notation $a^{i}=a+s^{i}(T-t)$, with $s^i=i$, $i\in\{-1,0,1\}$, $s^{01}=\frac{\gamma}{\gamma+\beta}$ (with the convention $\frac{\gamma}{\gamma+\beta}=K_2$ when $\gamma+\beta=0$) and $s^{-101}=s=2\gamma+\beta-1$ --- \ie{} the expected average taking the mass at various atoms into account. The function is depicted in Figure~\ref{fig:VF}, together with a candidate sample path.

\begin{figure}[ht]
  \centering
  \includegraphics[width=0.7\textwidth]{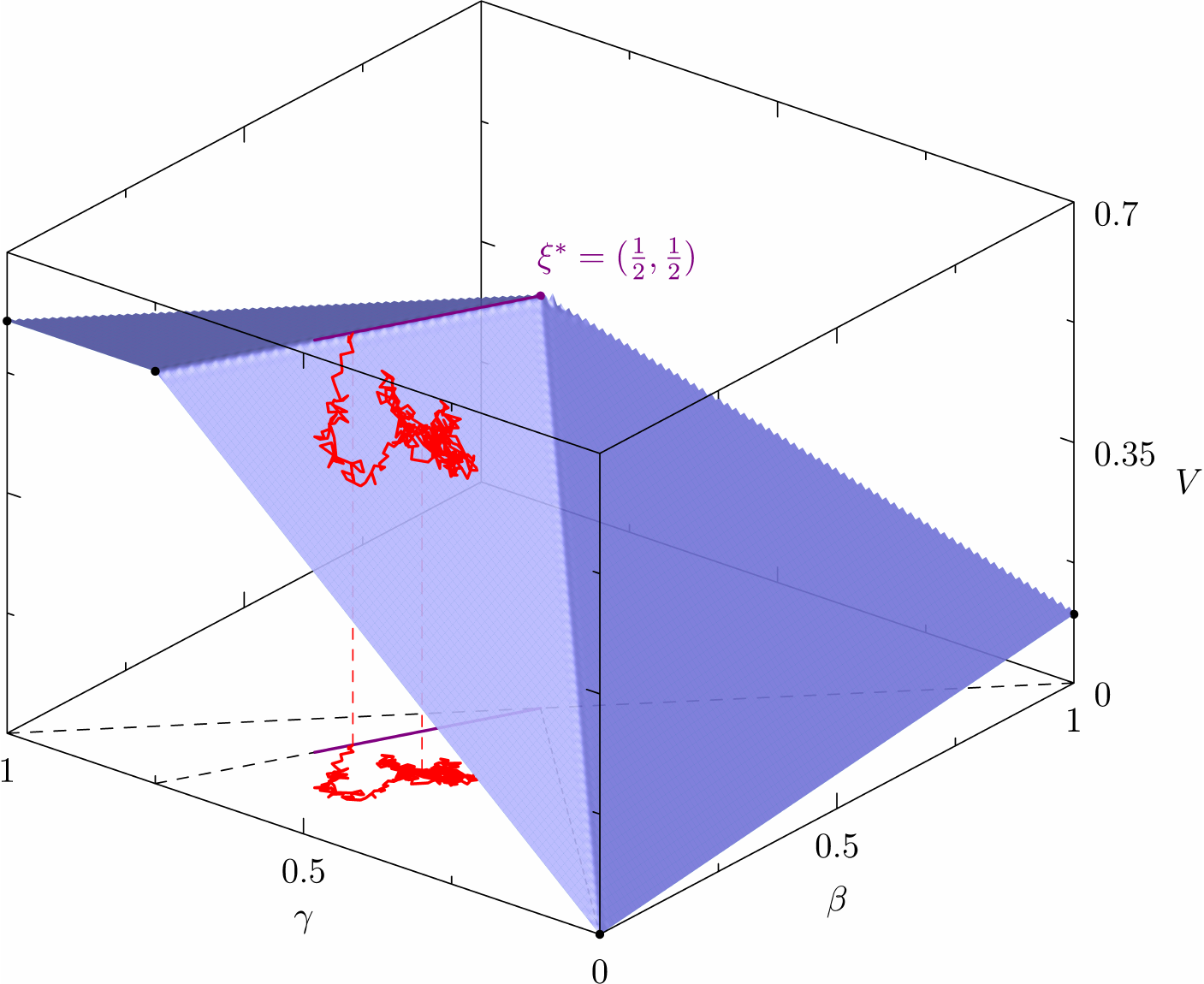}
  \caption{The value function graphed as a function of $\beta, \gamma$. Here $t=a=0, T=1$, $K_1 = -0.1, K_2 = 0.5$. Also shown (in red) is a possible path of $(\xi_r)$, starting from $(\beta, \gamma) = (\frac{1}{4},\frac{1}{2})$. The planar regions in the diagram correspond to the regions (i), (iii) and (iv) given in \eqref{eq:ex_v}. The process starts in region (iii), and runs until it hits the boundary of this region and region (i). The continuing path is then shown running along the boundary of regions (i) and (iii). In an optimal model, this behaviour happens at time 0, although note that there are many possible solutions: for example, the movement along the boundary between regions (i) and (iii) could happen at an time between $t=0$ and $t=T$. On reaching the point $\xi^*$, the process is unable to move any further before the time $t=T$ without being suboptimal.}
  \label{fig:VF}
\end{figure}

\begin{example}\label{ex:2}
  Observe that sending $K_2\to\infty$, $V(t,a;\beta,\gamma)$ reduces to the value-function for the (convex) payoff function $F(A_T)=(A_T-K_1)_+$ given in Section \ref{sec:convex} with $\mu$ given by \eqref{eq:mu_nonconvex}.  Alternatively, let $\beta=0$, $K_1=0$ and $K_2=\frac{1}{2}$. Then, $V(t,a;\gamma)$ reduces in the following way:
  \begin{equation} \label{eq:ex_dpe_2} 
    V(t,a;\gamma)=\left\{
      \begin{array}{lll}
        \frac{1}{2}, & 2\gamma-1> \frac{1/2-a}{T-t } \\
        \gamma\frac{1}{1+\frac{1/2-a}{T-t}}, & 2\gamma-1 \le \frac{1/2-a}{T-t}
      \end{array}\right. .
  \end{equation}
  Note that for $t=a=0$ and $2\gamma-1 \le 1/2$, $V_t+s V_a=0$ if and only if $\gamma=0$ or $\gamma=3/4$. Therefore the optimal model $(S_t)$ will jump to either $-1$ or $1/2$ at time $t=0$, and then stay constant until time $t=T$ where $\mu$ is embedded. 
\end{example}

It can be verified by elementary calculations that $V(t,a;\beta,\gamma)$ is continuous. The next result verifies that it is a (bounded) solution to equation \eqref{eq:dpp} with $F(a)$ and $\mu$ given by \eqref{eq:payoff_nonconvex} and \eqref{eq:mu_nonconvex}. Hence, according to Theorem \ref{thm:dpp}, $V(t,a;\beta,\gamma)$ is indeed the associated value function.

\begin{proposition}
  The function $V(t,a;\beta,\gamma)$ defined in \eqref{eq:ex_v} is the unique non-negative viscosity solution bounded by $K_2-K_1$, to the equation \eqref{eq:dpp} equipped with the boundary condition \eqref{eq:boundary_cond}. In particular, it is the value-function associated with the payoff \eqref{eq:payoff_nonconvex} and $\mu$ given by \eqref{eq:mu_nonconvex}.
\end{proposition}

\begin{proof}	
  Differentiating $V(t,a;\beta,\gamma)$ within the respective regions, we obtain that $V_t(t,a;\beta,\gamma)+sV_a(t,a;\beta,\gamma)=0$ in regions (i), (ii) and (vii), and
  \begin{equation} \label{eq:eq:diff} 
    V_t+sV_a
    \;=\;\left\{
      \begin{array}{lll}
        \frac{(K_2-K_1)(2\gamma+\beta)}{\left(1+\frac{K_2-a}{T-t }\right)^2}
        \left(-\frac{K_2-a}{(T-t )^2}+\frac{2\gamma+\beta-1}{T-t }\right)
        & \textrm{in (iii)}\\
        -\gamma+(2\gamma+\beta-1)(\gamma+\beta) & \textrm{in (iv)} \\
        -\gamma\frac{K_2-K_1}{(K_2-a)^2}
        \Big(K_2-a-(2\gamma+\beta-1)(T-t )\Big) & \textrm{in (v)} \\
        2\gamma\left(\gamma+\frac{\beta}{2}-1\right) & \textrm{in (vi)} 
      \end{array}\right. ,
  \end{equation}
  where $s=\vec{x}^\alpha\cdot(1-\beta-\gamma,\beta,\gamma)=2\gamma+\beta-1$.
  Using that $2\gamma+\beta-1\le \frac{K_2-a}{T-t }$ in regions (iii) and (v), and that $\gamma+\beta\le 1$ implies $(2\gamma+\beta-1)\frac{\gamma+\beta}{\gamma}\le 1$ for $\gamma>0$, it follows from \eqref{eq:eq:diff} that $V_t+sV_a\le 0$ within the respective regions. In consequence, with $V_{\vec{v}}$ denoting the directional derivative in the direction of $\vec{v}$, it holds on $\{\gamma+\beta\le 1\}$ that
  \begin{equation}
    \label{eq:direc_ineq} V_{\vec{v}}\le 0, \qquad
    \vec{v}=(1,2\gamma+\beta-1,0,0). 
  \end{equation}		
  
  Recall that $V(t,a;\pmb\xi)$ is a viscosity super (resp. sub) solution to \eqref{eq:dpp} if for each $\varphi\in \mathcal{C}^{1,1,2}$, and at each point $(\bar t ,\bar a,\pmb{\bar\xi})$ minimizing (resp. maximizing) $V-\varphi$,
  \begin{equation}	\label{eq:hjb_sub_super}
    \max\left\{\varphi_t+(2\bar\gamma+\bar\beta-1)
    \varphi_a, \frac{1}{2}\sup_{\vec{w}\in\mathbb{S}^2}
      \left[\tr\left(\vec{w}\vec{w}^\T D^2_{\pmb\xi} \varphi\right)\right]\right\} \le 0, \quad (\textrm{resp. $\ge 0$}).
  \end{equation}
  We first argue that $V$ is a sub solution. To this end, let $\varphi\in\mathcal{C}^{1,1,2.2}$ and $\vec{\bar z} = (\bar t, \bar a,\bar \beta,\bar \gamma)$ such that $\vec{\bar z}$ maximises $V-\varphi$.  Note that there exists $\vec{\bar w}\in\mathbb{S}^2$ such that the directional derivative at $\vec{\bar z}$ along $\vec{\bar w}$ (keeping $a$ and $t$ constant) satisfies $V_{\vec{\bar w}\vec{\bar w}}= 0$; if $\vec{\bar z}\in \{2\gamma+\beta-1=\frac{K_2-a}{T-t }\}$ or $\vec{\bar z}\in\{\frac{\gamma}{\gamma+\beta}=\frac{K_2-a}{T-t }\}$, let $\vec{\bar w}$ in the direction of that line.  Since $(V-\varphi)_{\vec{\bar w}\vec{\bar w}}\le 0$, it follows that $\tr\big(\vec{\bar w}\vec{\bar w}^\T D^2_{\pmb\xi} \varphi \big)\ge 0$ (note that $\tr\big(\vec{w}\vec{w}^\T D^2_{\pmb\xi} \varphi \big)=\vec{w}^\T D^2_{\pmb\xi} \varphi \;\vec{w}=\varphi_{\vec{w}\vec{w}}$). In consequence,
  \begin{equation} \label{eq:proof_visc1} 
    \sup_{\vec{w}\in\mathbb{S}^2}
    \left[\tr\left(\vec{w}\vec{w}^\T D^2_{\pmb\xi} \varphi \right)\right] \ge 0. 
  \end{equation}
  In order to show that $V$ is a super solution, let $\varphi\in\mathcal{C}^{1,1,2,2}$ and let $\vec{\bar z} = (\bar t,\bar a, \bar \beta,\bar \gamma)$ be a minimiser to $V-\varphi$. Due to the concavity of $V$ as a function of $\beta$ and $\gamma$ and the differentiability of $\varphi$, $\vec{\bar z}$ must lie strictly within one of the regions given in \eqref{eq:eq:diff}. Hence, for all $\vec{w}\in\mathbb{S}^2$, $\tr\big(\vec{w}\vec{w}^\T D^2_{\pmb\xi} (V-\varphi) \big)\ge 0$ and $\tr\big(\vec{w}\vec{w}^\T D^2_{\pmb\xi} V \big)=0$. In consequence,
  \begin{equation} \label{eq:proof_visc2} \sup_{\vec{w}\in\Sphere^2}
    \left[\tr\left(\vec{w}\vec{w}^\T D^2_{\pmb\xi} \varphi \right)\right] \le 0. 
  \end{equation} 
  Let now $\vec{v}=(1,2\bar \gamma+\bar \beta-1,0,0)$. Since $\vec{\bar z}$ minimises $V-\varphi$, it follows that $\varphi_{\vec{v}} \le V_{\vec{v}}$. According to \eqref{eq:direc_ineq}, we thus obtain
  \begin{equation*} 
    \varphi_t+ 
    (2\bar\gamma+\bar\beta-1)\;\varphi_a 
    ~\le~ V_{\vec{v}} ~\le~ 0, 
  \end{equation*}
  which combined with \eqref{eq:proof_visc2} renders \eqref{eq:hjb_sub_super}.
  
  It remains to argue the boundary conditions \eqref{eq:boundary_cond}. Note that for $t=T$, the only possible regions are (i), (ii) and (vii) (for $a\ge K_2$, $a\in[K_1,K_2)$ and $a< K_1$) for which $V(t,a;\beta,\gamma)$ is given, respectively, by $K_2-K_1$, $a-K_1$ and $0$. Hence, the terminal condition is satisfied. Next, note that for $\gamma=0$ and $\beta=0$ or $\gamma+\beta=1$, the problem reduces, respectively, to the convex case and the case presented in Example \ref{ex:2}. This verifies the second boundary condition and we conclude.
\end{proof}

We now discuss the optimal control associated with the value function \eqref{eq:ex_v} and the corresponding solution to the basic optimisation problem, Problem~\ref{prob:basic}.  Indeed, recall Lemma~\ref{lem:1}, which says that if Problem \ref{prob:MVM} admits an optimal solution, this solution corresponds to a solution of Problem~\ref{prob:basic}.  Naturally, the solution coincides with the conjectured optimal model used to deduce the form of $V(t,a;\beta,\gamma)$. However, our aim below is to illustrate how it may be deduced directly from the value function and, in consequence, from the dynamic programming equation \eqref{eq:hjb2}--\eqref{eq:cond_hjb} and to show that it is non-unique, and also non-trivial since it necessarily has a jump $t=T$ as well as $t=0$.  We let $T=1$ and split the behaviour into three parts.

(I)  \emph{Real time is kept constant while the measure-valued martingale evolves ($\lambda_u=0$ and $(\beta_u,\gamma_u)$ diffuses): $(S_t)$ jumps to certain points at time $t=0$.}
			
Depending on the parameters of the problem, the starting point $(0,0;\beta,\gamma)$ lies in one of the regions (i), (iii), (iv) or (v).  It follows from the DPP equation \eqref{eq:dpp}, that the model can evolve in real time only if $V_t+sV_a=0$. According to \eqref{eq:eq:diff}, while $V_t+sV_a=0$ for all $(\beta,\gamma)$ in regions (i), (ii) and (vii), it holds for the remaining regions that $V_t+sV_a=0$, if and only if,
\begin{equation}	\label{eq:ex_bound}
  \left\{
    \begin{array}{lll}
      2\gamma+\beta-1=\frac{K_2-a}{T-t} \quad\textrm{or}\quad (\beta,\gamma)=(0,0) & \textrm{in (iii)} \\
      \gamma+\beta=1 \quad\textrm{or} \quad (\beta,\gamma)=(0,0) & \textrm{in (iv)}\\
      2\gamma+\beta-1=\frac{K_2-a}{T-t} \quad\textrm{or} \quad \gamma = 0 & \textrm{in (v)} \\
      \gamma = 0 \quad \textrm{or}\quad (\beta,\gamma)=(0,1) & \textrm{in (vi)}
    \end{array}\right. .
\end{equation}
In consequence, if starting in region (i), one may immediately evolve in (real) time. However, if starting in regions (iii)--(v), (real) time must be kept constant while $(\beta_u,\gamma_u)$ are allowed to diffuse until the above boundaries are reached: that is, $\lambda_u=0$ until the measure-valued martingale $\xi_{\cdot} = (\beta_{\cdot},\gamma_{\cdot})$ satisfies \eqref{eq:ex_bound}.  Note that since $V(0,0;\beta_u,\gamma_u)$ is a martingale, if at the line $\frac{\gamma}{\gamma+\beta}=\frac{K_2-a}{T-t}$, diffusion will take place only in the direction of that line and the region remains the same until the boundaries are reached. This implies that the associated price process $(S_t)$ jumps to certain points at time $t=0$.

(II) \emph{Progress in real time only ($\lambda_u=1$): $(S_t)$ is kept constant.}

Once the measure-valued martingale satisfies \eqref{eq:ex_bound}, (real) time might start to evolve ($\lambda_u>0$). In particular, one might let $\lambda_u=1$ which implies that $(S_t)$ is kept constant. 	
On a case by case basis, it can be verified that once at a point where $V_t+sV_a=0$, this remains the case. For example, consider sitting at $\{\gamma+\beta=1\}$ in region (iv). With a slight abuse of notation, we see that at this line   
\begin{equation} \label{eq:evolution_time_2}
  \pd{}{t}a^{i}(t)=
  \left\{
    \begin{array}{lll}
      \pd{}{t} a(t)~=~2\gamma+\beta-1 ~\ge~0,& i=0\\
      \pd{}{t} \big(a(t)-(T-t)\big)=2\gamma+\beta~\ge~ 0,& i=-1\\
      \pd{}{t} \big(a(t)+\frac{\gamma}{\gamma+\beta}(T-t)\big)
      ~=~ 0	,& i=1
    \end{array}\right. . 
\end{equation}
Hence, when evolving in (real) time, and with no change in $\beta, \gamma$, $a^{-1}(t)=a(t)-(T-t)$ may move above $K_1$ and so the point $(t,a(t);\beta,\gamma)$ moves to region (ii). Since $V_t+sV_a=0$ within region (ii), the claim holds for this case.  Similar arguments apply to the other cases.

(III) \emph{When $T_u=T$, the measure-valued martingale $\xi_\cdot = (\beta_\cdot,\gamma_\cdot)$ terminates: $(S_t)$ jumps and embeds $\mu$ at $t=T$.}
		
Real time may run until $T_u=T$. Thereafter $\lambda_u=0$ and $(\beta,\gamma)$ diffuses until $\xi_u$ terminates; that is, until $\xi_\cdot =(\beta_\cdot,\gamma_\cdot)$ reaches $(0,0)$, $(0,1)$ or $(1,0)$. As expected, $V(T,a(T);\beta,\gamma)$ stays constant during this process as it is independent of $\gamma$ and $\beta$. This step corresponds to $S$ embedding $\mu$ via a jump at $t=T$.

The evolution in time and measure described in (II) and (III) could, partially, happen simultaneously or in the reverse order. This implies that the optimal model is not unique.  For example, having reached the line $\gamma+\beta=1$ in region (iv), one might let $(\beta_u,\gamma_u)$ continue to diffuse until reaching either of the points $(1-\frac{K_2-A_u}{T-u},\frac{K_2-A_u}{T-u})$ or $(1,0)$, before letting (real) time evolve. This corresponds to the behaviour used to deduce $V(t,a;\beta,\gamma)$: \ie{} $(S_t)$ jumps to one of the values $-1$, $0$ or $K_2$ at time $t=0$.  Alternatively, by letting $\lambda_u\in(0,1)$, and supposing that $A_t \ge K_1$, one may let time and measure evolve simultaneously, which corresponds to $Y_u = A_u + (T-u)S_u$, $u\in(t,1)$ being either constantly equal to $-1$ or moving (as a continuous martingale) between the values $0$ and $K_2$. Observe that this behaviour may result in a different distribution to $A_T$ in comparison with the case where all the diffusion happens immediately.  Similar behaviour can be observed in the regions (i), (ii) and (v), although the distribution of $F(A_T)$ then remains the same.
     
While the optimal model is not unique, we note that it has certain characteristics: the model necessarily has a jump at both $t=0$ and $t=T$. Indeed, there is a certain amount of mass which is `locked in', and cannot be embedded until the terminal time $t=T$. This to ensure that $S_t=K_2$, $t\in(0,T)$, with a certain probability (\eg{} in regions (iii) and (v)). On the other hand, sending $K_2\to\infty$ and thus isolating the behaviour at the convex kink $K_1$, we see that the mass terminating above/below $K_1$ must already at time $t=0$ be accumulated above/below $K_1$. Although affected by the presence of $K_2$, this feature is present also for the general case (\eg{} regions (iv) and (vi)).
	 
\section{Conclusions and Further Work}\label{sec:concl-furth-work}     

In this paper we have considered the model-independent pricing problem for Asian options using a novel approach based on measure-valued martingales. While this paper concentrated on the case of Asian options, the main ideas should generalise to other cases, and may provide insights beyond the existing literature. Moreover, there are a number of natural questions which arise from our work:
\begin{itemize}
\item Is it possible to generalise the results in this paper to the case of a general starting law? Financially, this has the interpretation of pricing a forward starting option at time $0$, where $0 < T_0 < T_1$, the option pays the holder the amount $F\left(\int_{T_0}^{T_1} S_u \, du\right)$ at time $T_1$ and the price of call options are known at times $T_0$ and $T_1$. Write $\lambda$ for the implied law of $S_{T_0}$ and $\mu$ for the implied law of $S_{T_1}$. It follows immediately from the results of this paper that the problem is equivalent to finding a function $m:\R \to \Pc^1, x \mapsto m_x$ which maximises $\int U(0,m_x,0) \, \lambda(\dx)$ over all functions $m$ such that $\int m_x(A) \, \lambda(\dx) = \mu(A)$, for all Borel sets $A$ and $x = \int y \, m_x(\dy)$. However, it would be interesting to have a dynamic formulation of the problem, similar to the simple case.
\item The PDE \eqref{eq:dpp} is formulated for the case of atomic measures. Is there a similar formulation that holds when $\xi$ is only assumed to be measure-valued?
\item What is the corresponding formulation for \eqref{eq:dpp} in the case of (say) options on variance?
\item Do the methods described above extend to related problems in higher dimensions? If the formulation is given in terms of a measure $\mu$ on $\R^d$, one might hope that a very similar approach would be possible. Is this also true of (the financially more meaningful) case where $S_t \in \R^d$, and the marginal distribution of each component of $S_T$ is specified?
\end{itemize}

\appendix
\section{A Formal Dynamic Programming Principle}
In this section, we formally derive a Dynamic Programming Principle (DPP) for the pricing problem in its weak form given in Definition~\ref{prob:MVM}. We note that our previous results do not make use of this DPP, but we believe that this result is of independent interest. We choose to follow closely the setup in \cite{Zitkovic:2014aa}; see however also \cite{El-Karoui:2013ab} and \cite{Nutz:2013aa} for similar arguments.

We denote by $\mathbb{D}$ the set of c\`adl\`ag paths on $[0,\infty)$ taking values in $E:=\Pc(\R)\times[0,T]\times\R$, where we equip $\Pc(\R)$ with the topology induced by the $\mathcal{W}_1$-metric and $E$ with the product topology; in particular, this renders $E$ a Polish space, and using the Skorokhod topology on $\mathbb{D}$ it is a Polish space too.  For $x,x'\in E$, we write $d(x,x'):=\Wc_1(\xi,\xi') \vee |t-t'| \vee |a-a'|$.  A generic path in $\mathbb{D}$ is denoted by $\omega$ and we use $X=(\xi,T,A)$ for the co-ordinate process: $X_r(\omega)=(\xi_r,T_r,A_r)(\omega)=\omega(r)$.

The set of all probability measures on $\mathcal{B}(\mathbb{D})$ is denoted by $\mathfrak{P}$. A map $\nu:E\times \mathcal{B}(\mathbb{D})\to[0,1]$ is called a (universally) measurable kernel if  i) $\nu(x,\cdot)\in\mathfrak{P}$ for all $x\in E$, and ii) $E\ni x\to\nu(x,A)$ is (universally) measurable for all $A\in\mathcal{B}(\mathbb{D})$; recall that the universal $\sigma$-algebra is the intersection of the completions of the Borel $\sigma$-algebra over all probability measures on the space, and that universally measurable functions are integrable with respect to any such probability measure.  We write $\nu_x$ for the probability measure $\nu(x,\cdot)$ and interpret $\nu$ as a (universally) measurable map $E\to\mathfrak{P}$.

A Borel-measurable map from $\mathbb{D}$ to $[0,\infty)$ is called a random time.  For any random time $\tau$, we define the shift-operator $\theta_\tau$ on $\mathbb{D}$ by $X_r(\theta_\tau(\omega))=X_{\tau(\omega)+r}(\omega)$.  Further, for a random time $\tau$ and any two paths $\omega,\omega'\in\mathbb{D}$ such that $X_\tau(\omega)=X_0(\omega')$, the concatenation $\omega*_\tau\omega'$ is an element of $\mathbb{D}$ specified by 
\begin{equation*}
  X_t(\omega*_\tau\omega')
  =
  \ind_{\{t<\tau(\omega)\}} X_t(\omega)
  +
  \ind_{\{t\ge \tau(\omega)\}} X_{t-\tau(\omega)}(\omega').
\end{equation*}
For a random time $\tau$, a probability measure $\mu\in\mathfrak{P}$ and a universally measurable kernel $\nu$, we then define the concatenation $\mu*_\tau\nu$ as the probability measure in $\mathfrak{P}$ given by 
\begin{equation*}
  (\mu*_\tau\nu)(A)=
  \iint\ind_A(\omega*_\tau\omega')\nu_{X_\tau(\omega)}(\di\omega')\mu(\di\omega), 
  \qquad A\in\mathcal{B}(\mathbb{D}). 
\end{equation*}

We let $\mathbb{F}^0=\{\Fc^0_r\}_{r\in[0,\infty)}$ denote the filtration generated by the co-ordinate process $X$, and let $\mathbb{F}=\{\Fc_r\}_{r\in[0,\infty)}$ be its right-continuous hull; i.e. $\Fc_r=\cap_{s>r}\Fc^0_s$, for $r\ge 0$.  For $x=(\xi,t,a)\in E$, we denote by $\Pc_x$ the set of probability measures in $\mathfrak{P}$ for which:
\begin{enumerate}[label=(\roman*)]
\item\label{item:1} $X_0=x$ a.s., 
\item\label{item:2} $\xi_r$ is a measure-valued $\mathbb{F}$-martingale,
\item\label{item:3} $T_r$ is non-decreasing with $\lim_{r\to\infty}T_r=\infty$, a.s.,
\item\label{item:4} $A_r=a+\int_0^{r\wedge\tau_0} \bar\xi_{u-}\di T_u$ a.s., where $\bar\xi_\cdot=\int x\xi_\cdot(\dx)$ and $\tau_0=\inf\{r:T_r\ge T\}$.
\end{enumerate}
Finally, we note that according to Lemma 3.12 in \cite{Zitkovic:2014aa}, there exists a measurable functional $\bar X=(\bar\xi,\bar T,\bar A):\mathbb{D}\to E$ such that $\bar X(\omega)=\lim_{t\to\infty}X_t(\omega)$ whenever the limit exists and $\bar X(\theta_t(\omega))=\bar X(\omega)$ for all $t\ge 0$. We let $G(\omega):=F(\bar A(\omega))$; for any $\mu\in\Pc_x$, $x\in E$, we then have that $G=\lim_{t\to\infty}F(A_t)$ a.s.  We define the problem: 
\begin{equation}\label{eq:vdefn}
  v(x)=\sup_{\mu\in\Pc_x}\E^\mu\left[G\right]. 
\end{equation}

\begin{lemma}
  The value function defined in \eqref{eq:vdefn} coincides with the value function given in \eqref{eq:Udefn}. In particular, $x\mapsto v(x)$ is continuous.  
\end{lemma}

\begin{proof}
  Let $(\Omega, \Gc, (\Gc_r), \P, (\xi_r), (\lambda_r))$ be a multiple as specified in Problem \ref{prob:MVM}. Without loss of generality, let $x=(\mu,0,0)$ (Note that in Problem~\ref{prob:MVM} it was assumed only that $\xi_{0-} = x$; by considering a time transformation $t \mapsto \frac{(t-\eps)_+}{T-\eps}T$, this difference can be seen to be irrelevant). 
  Since any martingale is a martingale in its own filtration, it follows that any such multiple induces on the canonical space $\mathbb{D}$ a measure $\mu\in\Pc_x$.
  Conversely, any probability measures $\mu\in\Pc_x$ together with the space $(\mathbb{D},\mathcal{B}(\mathbb{D}),\mathbb{F})$ and the canonical process $(\xi,T)$ produces such a multiple. Indeed, the fact that one may, without loss of generality, assume that $T_\cdot$ is absolutely continuous a.s. follows as in the proof of Lemma~\ref{lem:1}. Moreover, since $\tau_0<\infty$ a.s., for any pair $(T_r,\xi_r)$ with $\xi$ a measure-valued martingale, one may construct a terminating measure valued martingale $\tilde\xi_r$ such that $(T_r,\tilde\xi_r)$ yields the same value of the payoff.
  The continuity is then an immediate consequence of Lemma \ref{lem:Ucontinuity}.	 
\end{proof}

\begin{remark}
  Let $\mathbb{\tilde F}^0$ be the filtration generated by $(T,\bar\xi)$, where $\bar\xi_\cdot=\int x\xi_\cdot(\dx)$, and let $\mathbb{\tilde F}$ be its right-continuous hull. Further, let $\tilde \Pc_x$ denote the set of measures in $\mathfrak{P}$ which satisfy properties \ref{item:1} to \ref{item:4} above with the difference that $\xi$ is only assumed to be a measure-valued $\mathbb{\tilde F}$-martingale. We then have that
  \begin{equation} 
    v(x)
    =
    \sup_{\mu\in\tilde\Pc_x}\E^\mu\left[G\right].
  \end{equation}
  Indeed, this follows from the proof of Lemma~\ref{lem:1}, where the constructed measure-valued martingales are indeed adapted to the filtration generated by $T_\cdot$ and $\bar\xi_\cdot$.  \end{remark}

We are now ready to state the DPP. For simplicity we provide it here for bounded payoff functions. We denote by $\mathcal{T}$ the set of finite $\mathbb{F}$-stopping times.

\begin{theorem}
  Let $F:\R_+\to\R_+$ be bounded and Lipschitz. Then, for all $x\in E$ and $\tau\in\mathcal{T}$,
  \begin{equation*}
    v(x)=\sup_{\mu\in\Pc_x}\E\left[v(X_\tau)\right].
  \end{equation*}
\end{theorem}

\begin{proof}
  Given $\varepsilon>0$, $x\in E$ and $\tau\in\mathcal{T}$, take $\mu\in\Pc_x$ such that $v(x)-\varepsilon\le \E^\mu[G]$. Let $\nu_x$ be the regular conditional probability distribution of $\theta_\tau$ under $\mu$ given $X_\tau=x$; since $\mathbb{D}$ is Polish it exists $\mu\circ X^{-1}_\tau$-a.s. Recall that for any $f\in C_b(\R_+)$, $\xi_\cdot(f)$ is a bounded $\mu$-martingale. By use of the same argument as given in the proof of Proposition 3.11 in \cite{Zitkovic:2014aa}, we may then conclude that $\xi_\cdot(f)$ is a $\nu_x$-martingale for $\mu\circ X_\tau^{-1}$-almost all $x\in E$. It follows that $\nu_x\in\Pc_x$ for $\mu\circ X_\tau^{-1}$-almost all $x\in E$. As argued in the proof of Proposition 2.5 in \cite{Zitkovic:2014aa}, we may further pick a universally measurable version of $\nu_x$ such that $\nu_x\in\Pc_x$ for all $x\in E$.
  Now, note that $G(\omega')=G(\omega*_\tau\omega')$ for all $\omega,\omega'\in\mathbb{D}$ with $X_\tau(\omega)=X_0(\omega')$, and thus
  \begin{eqnarray*}
    \iint G(\omega')\nu_{X_\tau(\omega)}(\di\omega')\mu(\di\omega)
    =\iint G(\omega*_{\tau}\omega')\nu_{X_\tau(\omega)}(\di\omega')\mu(\di\omega).
  \end{eqnarray*}	
  By use of the properties of the r.c.p.d., we thus obtain the following line of equalities: 	
  \begin{eqnarray*}
    \E^\mu[G]
    =\E^\mu[G\circ\theta_\tau]
    =\E^{\mu*_\tau\nu}[G]
    =\E^\mu[g(X_\tau)],
  \end{eqnarray*}	
  where $g(x)=\E^{\nu_x}[G]=\int G(\omega')\nu_x(\di\omega')$. 
  Hence, $v(x)-\varepsilon\le \E^\mu[v(X_\tau)]$ for some $\mu\in\Pc_x$, and since $\varepsilon$ was chosen arbitrarily we obtain $v(x)\le\sup_{\mu\in\Pc_x}\E\left[v(X_\tau)\right]$.

  In order to argue the reverse inequality, for any $\varepsilon>0$, we first argue the existence of a measurable kernel $\nu$ with $\nu_x\in\Pc_x$ and $ \E^{\nu_x}[G]\ge v(x)-\varepsilon$, for each $x\in E$. To this end, we define a mapping $E\times\Db\ni(x,\bar\omega)\mapsto\omega^{x,\bar\omega}\in\Db$ such that for each $x=(\xi,t,a)\in E$, the mapping $\bar\omega=(\bar\xi_{\cdot},\bar t_{\cdot},\bar a_{\cdot})\mapsto\alpha^x(\bar\omega):=\omega^{x,\bar\omega}=(\xi_{\cdot},t_{\cdot},a_{\cdot})$ modifies the path $\bar\omega$ as follows:
  \begin{equation}\label{eq:path_modification}
    \left\{
      \begin{array}{cll}
        \xi_r(\di y) &=& \int \bar\xi_r(\dx)m^\xi(x,\di y),\\
        t_r &=& \bar t_r + t - \bar t_0, \\
        a_r &=& a + \int_0^{r\wedge\tau_0} \int x\xi_{u-}(\dx) \dt_u,				
      \end{array}
    \right.
  \end{equation}
  where the family $m^\xi(\cdot,\dy)$ satisfies $\Wc_1(\bar\xi_0,\xi)=\iint|x-y|\bar\xi_0(\dx)m^\xi(x,\dy)$, and $\tau_0=\inf\{r:t_r=T\}$.  Then $(x,\bar\omega)\mapsto\omega^{x,\bar\omega}$ is $\Bc(E)\times\Bc(\Db)$ measurable. 
  Hence, for any $\bar\P\in\cup_{x\in E}\Pc_x$, defining $\bar\nu_x:=\bar\P\circ (\alpha^x)^{-1}$, $x\in E$, yields a measurable kernel $\bar\nu$ with $\bar\nu_x\in\Pc_x$. Indeed, the martingale property of $\xi$ under $\bar\nu_x$ follows as in the proof of Lemma \ref{lem:Ucontinuity}. 
  Further, from \eqref{eq:path_modification} we have that
  \begin{align*}
    \left|\bar a_\infty-a_\infty\right| 
    &\le \int_{\bar t_0}^{T}
           \left|\int x\bar\xi_{\bar t^{-1}_s}(\dx)-\int x \xi_{\bar t^{-1}_s}(\dx)\right|\di s\\
    & \qquad {} +\int_{T-(t-\bar t_0)}^T\left(\int x\bar\xi_{\bar t^{-1}_s}(\dx)
       +\int x\xi_{\bar t^{-1}_s}(\dx)\right)\di s + |\bar a_0-a|.
  \end{align*}
  Proceeding as in the proof of Lemma \ref{lem:Ucontinuity}, for any $\bar\P\in\Pc_{\bar x}$ with $\bar x=(\bar\xi,\bar t,\bar a)$, we then have that
  \begin{align*}
    \E^{\bar\P}\big[\big|A_\infty(\bar\omega)-  A_\infty(\alpha^x\circ \bar\omega)\big|\big] 
    \le T\;\Wc_1(\bar\xi,\xi) + (t-\bar t)\int x\bar\xi(\dx)\vee \int x\xi(\dx) + |\bar a-a|,
  \end{align*}
  and with $\bar\nu_x=\bar\P\circ (\alpha^x)^{-1}$, the Lipschitz property of $F$ thus yields $|\E^{\bar \P}[G]-\E^{\bar\nu_x}[G]|\le\delta^{\bar\xi}(d(\bar x,x))$ for some modulus of continuity $\delta^{\bar\xi}$. Now, let $\varepsilon>0$, and let $\{x^n\}_{n\in\mathbb{N}}$ be a countable dense subset of $E$. For each $n$, let $\P_n\in\Pc_{x^n}$ such that $\E^{\P_n}[G]\ge v(x^n)-\frac{\varepsilon}{3}$. Further, for each $x^n$, let $r_n$ such that for all $x\in B^n:=\{x\in E: d(x,x^n)\le r_n\}$, it holds that $v(x^n)\ge v(x)-\frac{\varepsilon}{3}$ and $\left|\E^{\P_n}[G]-\E^{\nu^n_x}[G]\right| \le \frac{\varepsilon}{3}$ with $\nu^n_x := \P_n\circ (\alpha^x)^{-1}$; the existence of such $r_n$, $n\in\Nb$, follows from the above and Lemma \ref{lem:Ucontinuity}. We then define the measurable kernel $(\nu_x)_{x\in E}$ by
  \begin{equation}\label{eq:def_nu}
    \nu_x:=\sum_{n\in\Nb}\ind_{C^n}(x)\;\P_n\circ (\alpha^x)^{-1}, 
    \quad\textrm{where}\quad 
    C^n=B^n\setminus\bigcup_{k-1}^{n-1}B^k.
  \end{equation}
  By construction, for $x\in C^n$, $n\in\Nb$, we then have that 
  \begin{equation*}
    \E^{\nu_x}[G]
    \ge \E^{\P_n}[G]-\frac{1}{3}\varepsilon
    \ge v(x^n)- \frac{2}{3}\varepsilon
    \ge v(x) - \varepsilon. 
  \end{equation*}
  Hence, $\nu$ is a measurable kernel with $\nu_x\in\Pc_x$ and $\E^{\nu_x}[G]\ge v(x)-\varepsilon$, for $x\in E$.

  In order to conclude, we take $x_0\in E$, $\mu\in\Pc_{x_0}$, $\tau\in\mathcal{T}$ and $\nu$ as constructed in \eqref{eq:def_nu}. Since $\xi_\cdot(f)$ is a bounded $\mu$-martingale for any $f\in C_b(\R_+)$, we may use the same arguments as in the proof of Proposition 3.10 in \cite{Zitkovic:2014aa} to deduce that $\xi_\cdot(f)$ is also a $\mu*_\tau\nu$-martingale. We may thus conclude that $\mu*_\tau\nu\in\Pc_{x_0}$. Letting $g(x)= \E^{\nu_x}[G]$ and noticing that $g$ is measurable, we thus obtain
  \begin{equation*}
    v(x_0)
    \ge \E^{\mu*_\tau\nu}[G]
    = \E^\mu[g(X_\tau)]
    \ge \E^\mu[v(X_\tau)]-\varepsilon. 
  \end{equation*}		
  Since $\varepsilon$ and $\mu\in\Pc_{x_0}$ were both chosen arbitrarily, we obtain $v(x)\ge\sup_{\mu\in\Pc_x}\E\left[v(X_\tau)\right]$ and conclude.
\end{proof}
     
The above proof exploits the continuity properties of our problem in order to construct an approximately optimal measurable kernel; see \cite{Aksamit:2016aa} and \cite{Bouchard:2011aa} for similar approaches.

\renewbibmacro*{in:}{}

\printbibliography

\end{document}
